\documentclass[12pt,reqno]{amsart}
\usepackage{amssymb}
\usepackage{amsfonts}
\usepackage{indentfirst}
\usepackage{amsopn}
\usepackage{amsmath}
\usepackage{amsthm}
\usepackage{amscd}
\usepackage[dvipdf]{graphicx}
\usepackage[dvips]{epsfig}
\usepackage{color,xcolor}

\makeatletter\@addtoreset {equation}{section}\makeatother

\newtheorem {thm}{Theorem}[section]
\newtheorem {lem}[thm]{Lemma}

\newtheorem {cor}[thm]{Corollary}
\theoremstyle{definition}

\theoremstyle{assumption}
\newtheorem {assume}[thm]{Assumption}

\theoremstyle{observation}

\theoremstyle{remark}
\newtheorem{rem}[thm]{Remark}
\theoremstyle{example}
\newtheorem{ex}[thm]{Example}

\def\p{\partial}

\def\E{\operatorname{\mathbb E}}
\def\P{\operatorname{\mathbb P}}

\def\R{{\mathbb R}}
\def\N{{\mathbb N}}
\def\Z{{\mathbb Z}}
\def\T{{\mathbb T}}

\def\cX{\mathcal{X}}

\def\lbl{\label}
\def\be{\begin{equation}}
\def\ee{\end{equation}}
\def\p{\partial}

\def\1{\operatorname{\mathbf 1}}

\def\XXint#1#2#3{{\setbox0=\hbox{$#1{#2#3}{\int}$ }
\vcenter{\hbox{$#2#3$ }}\kern-.6\wd0}}

\def\tl{\tilde}
\newcommand{\diag}{{\operatorname{diag}}}
\newcommand{\bA}{\mathbf{A}}
\newcommand{\cut}[1]{\left\Vert#1\right\Vert_{\infty\to 1}}

\newcommand{\iu}{{i\mkern1mu}}

\begin{document}

\title[Turing bifurcations on deterministic and random graphs]{Turing bifurcation in the Swift--Hohenberg
  equation on deterministic and random graphs}

\author[G. S. Medvedev]{Georgi S. Medvedev}
\address{Department of Mathematics, Drexel University, 3141 Chestnut Street, Philadelphia, PA 19104}
\email{medvedev@drexel.edu}
\author[D. E. Pelinovsky]{Dmitry E. Pelinovsky}
\address{Department of Mathematics and Statistics, McMaster University, Hamilton, Ontario, Canada, L8S 4K1}
\email{dmpeli@math.mcmaster.ca}

\maketitle

\begin{abstract}
  The Swift--Hohenberg equation (SHE) is a partial differential equation 
  that explains how patterns emerge from a spatially homogeneous state.
  It has been widely used in the theory of pattern formation.
  Following a recent study by Bramburger and Holzer \cite{BraHol23}, we consider discrete SHE on
  deterministic and random
  graphs. The two families of the discrete models share the same continuum limit in the form of a nonlocal
  SHE on a circle. The analysis of the continuous system, parallel to the analysis of the classical SHE,
  shows bifurcations of spatially periodic solutions at critical values of the control parameters. However,
  the proximity of the discrete models to the continuum limit does not guarantee that the same bifurcations
  take place in the discrete setting in general, because some of the symmetries of the continuous model
  do not survive discretization.

  We use the center manifold reduction and normal forms to obtain precise information about the number
  and stability of solutions bifurcating from the homogeneous state in the discrete models on deterministic
  and sparse random graphs. Moreover, we present detailed numerical results for the discrete SHE on
  the nearest-neighbor and small-world graphs.
\end{abstract}

\section{Introduction}\label{sec.intro}
\setcounter{equation}{0}

Networks of interacting dynamical systems (a.k.a. interacting particle systems) form an important class of models of
natural and technological systems.
Examples include neuronal networks, swarm of fireflies, coupled
lasers, and power grids, to name a few \cite{DorBul12,KurPikRos,PG16,Str-Sync}.

When the number of particles is large, taking a limit as the number of particles
goes to infinity is an effective tool for analyzing network dynamics. The continuum limit in the form of a
nonlocal PDE has been very successful
for studying synchronization and pattern formation in large systems of coupled oscillators on a variety
of graphs \cite{Med14c,MedTan15b,MedTan18,WilStr06}. Therefore, it is important to understand the accuracy
of approximation of large dynamical networks
by a continuum equation. For solutions of initial value problems on convergent families of graphs, this was
done in \cite{Med14a, Med14b, Med19}. 

In many theoretical studies as well as in practical applications,
valuable information about system dynamics is gained by studying regimes bifurcating from simpler solutions under the variation of the control parameter. Therefore, it is of interest to understand
how well and under what conditions, the bifurcation structure of large networks can be obtained from its continuum limit. Specifically, suppose the continuum model undergoes a bifurcation at a certain value
of the control parameter. What can be said about the discrete system? Will it undergo a bifurcation at
a close parameter value? Will the bifurcating solutions resemble those obtained in the continuum case?
These are nontrivial questions in general, because certain features (such as symmetries) that are present
in the continuum limit may  not survive discretization. 
The questions become even more challenging if the large network system is random. 

In this paper, following the work of Bramburger and
Holzer \cite{BraHol23}, we address these questions in the context of the Turing bifurcation in the discrete Swift--Hohenberg equation (SHE) on deterministic (Cayley) and random graphs. We postpone the discussion of how our approach and results differ
from \cite{BraHol23} and turn to  the formulation of the
discrete and continuous models next.

The classical SHE  plays a prominent role in the theory of pattern formation (see \cite{Collet} and references therein). It has the following
form
\be\lbl{class-she}
\p_t u = - \left( 1 +\p^2_{x}\right)^2 u + \gamma u -u^3, \quad x\in \T\doteq \R/\Z,
\ee
where $u(t,x)\in\R$ and $\gamma$ is a control parameter. The normal form reduction near the bifurcation,
which makes use of the symmetries present in the system, shows existence of a two-parameter family of 
stable stationary nontrivial solutions
bifurcating from the trivial solution $u \equiv 0$ (see Section 2.4.3 in ~\cite{HaragusIooss}):
\begin{equation}\label{bifurcating-nonSHE}
u_{\gamma,\delta}(x) = 2\sqrt{ \frac{\gamma}{3}} \cos(x + \delta) +O(\gamma^{3/2})
\end{equation}
for small positive $\gamma$ and every $\delta\in \T$.

The nonlocal SHE is obtained by replacing the second derivative $\p^2_{x}$ with a nonlocal
operator $L_W$:
\begin{equation}\label{nSHE}
\p_t u= -\left(L_W -\kappa\right)^2 u +\gamma u -u^3,\quad x\in \T,
\end{equation}
where
\begin{align}\lbl{nLap}
  (L_W f)(x) &= \int_\T W(x,y) \left[ f(y)-f(x) \right] dy =: (K_W f)(x)- d_W(x) f(x).
\end{align}
Here, $W: \T^2\rightarrow [0,1]$ is a measurable function that
is symmetric a.e. on $\T^2$, and $f\in L^1(\T)$.

In this paper, we assume that $W$ has the following form
\begin{equation}\label{W-Cay}
  W(x,y)=S(x-y)
\end{equation}
for a given even function $S\in L^1(\T)$. In this case, $d_W$ is independent of $x \in \mathbb{T}$.
Graphons of the form \eqref{W-Cay} arise as limits of convergent sequences of  Cayley graphs. For this reason,
they are referred to as {\bf Cayley graphons}. It is instructive to study
this case first, because the nonlocal SHE \eqref{nSHE} with \eqref{nLap}  and \eqref{W-Cay} is the closest nonlocal analog of the classical model \eqref{class-she} on a periodic domain. Since $K_W$ has a discrete (real) spectrum with the only accumulation point at zero, the bifurcation of stable stationary nontrivial solutions occurs if $\kappa$ is set to 
$$
\kappa = -d_W + \lambda_{k_0}, \quad k_0 \in \mathbb{N},
$$
where $\{ \lambda_k \}_{k \in \mathbb{Z}}$ are eigenvalues of $K_W$ satisfying 
$\lambda_{-k}= \lambda_k$, $k \in \mathbb{N}$. Similar to (\ref{bifurcating-nonSHE}), we are interested in solutions bifurcating from $u = 0$ for small $\gamma$. The normal form analysis of (\ref{nSHE}) closely resembles the normal form analysis of (\ref{class-she}) and illustrates the role of translation invariance in the continuum model, see Theorem \ref{theorem-SH-local} below.

Along with the nonlocal model \eqref{nSHE} we consider a discrete model on a graph $\Gamma$ given by 
\begin{equation}\lbl{discrete-she}
  \dot u = -\left( L_{\Gamma}  -\kappa\right)^2 u +\gamma u -u^3,
  \end{equation}
  where $u(t) \in \R^n$, the cubic nonlinearity $u^3$ is understood in the componentwise sense, and the graph Laplacian $L_{\Gamma}$ is defined as follows
\begin{equation}\label{def-L_N-intro}
(L_{\Gamma} u)_i = \frac{1}{n} \sum_{j=1}^n a_{ij} \left( u_j-u_i\right),
\end{equation}
which is associated with the adjacency matrix $A_{\Gamma} = (a_{ij})_{1 \leq i, j \leq n}$.  

In this paper, we study \eqref{discrete-she} with \eqref{def-L_N-intro}
on two families of graphs, one is deterministic and the other is random.  Both families are constructed using
a given $W$ and converge to $W$ almost surely.
We refer the reader to \cite{Med14a, Med19} for the overview of the ideas from the graph limit theory that are relevant
here and to the references in the literature on this subject. For either graph model, one can show that the
solution of the IVP for \eqref{nSHE} approximates the solutions of the IVPs of the discrete problems on
finite time intervals (cf.~\cite{Med14a, Med14b, Med19}). Thus, \eqref{nSHE} provides a common continuum
limit for the discrete models on both deterministic and random graphs. 
Note that this however does not guarantee
that the solutions bifurcating from the trivial solution will resemble \eqref{bifurcating-nonSHE}. In fact,
in general, it is not clear that the discrete models will undergo a bifurcation at all. The reason for this is that the translation symmetry present in the continuum model does not survive discretization. The lack of the translation symmetry affects computations of the normal forms and thus the bifurcations. 
Our analysis predicts the exact number of solution families bifurcating from the trivial solution in the deterministic model, see Theorem \ref{theorem-SH-discrete}, based on the discrete group of symmetries. 
On the other hand, the symmetries are broken in the random model and we 
are only able to find the lower and upper bounds on the number of solution families bifurcating from the trivial solution, see Theorem \ref{thm.rand}.

The organization of the paper is as follows. In Section \ref{sec.discrete}, we complete the description of the discrete model \eqref{discrete-she} with \eqref{def-L_N-intro} by providing the details on the deterministic and random graphs.  The former are weighted Cayley graphs and the latter are W-random graphs (cf.~\cite{LovSze06}).

After that we turn to the analysis of the continuous model in Section~\ref{sec.continuum}.
Specifically, we demonstrate the well-posedness of the continuous model 
and analyze the bifurcation of the spatially homogeneous solutions. 
The bifurcation analysis is based on the normal form of the system in the Fourier coordinates.
The derivation of the normal form uses the symmetries present in the system and resembles
the analysis of the classical SHE in \cite{HaragusIooss}. The existence of a family of spatially periodic solutions is invariant with respect to the continuous spatial translations.

In Section~\ref{sec.Cayley}, we study discrete SHE on weighted Cayley graphs. The bifurcation
analysis follows the same lines as in the continuous case, where in place of  Fourier transform
we now use the discrete Fourier transform. An important distinction of the discrete model
is that the invariance under any translations in the continuous system is replaced by the
invariance with respect to a discrete group of translations. This results in 
finitely many solutions
bifurcating from the trivial solution instead of the continuous 
family (\ref{bifurcating-nonSHE}).

In Section~\ref{sec.random}, we move on to study the discrete model on random graphs, which is the main problem
addressed in this paper. In the random setting, the discrete model does not possess the discrete
translational invariance and the normal form approach, which worked  for Cayley
graphs, is no longer applicable. To overcome this obstacle, we use the proximity of the system on a sufficiently large
random graph to that analyzed in Section~\ref{sec.Cayley} to derive a leading order approximation for the normal form
of the system at hand. This allows us to study the bifurcation for the system on random graphs and to obtain precise bounds
on the number of solution families bifurcating in each random realization of the discrete graph. This is where our
approach is different from the approach taken in \cite{BraHol23}. The normal form for the bifurcation on a random graph is
compared with the one on the associated discrete deterministic graph instead of the one on the associated continuous
nonlocal model. We determine the translational parameter precisely from the normal form, whereas the 
translational parameter is defined implicitly in \cite[Theorem 4.1]{BraHol23} from a linear transformation of eigenvectors
of the linearized equations. Additional comments on the differences between our work and  \cite{BraHol23}
can be found in Remark \ref{rem-BraHol}.

In Section~\ref{sec.small-world}, we present numerical experiments with SHE on small-world graphs.
This example illustrates the selection 
of random stationary patterns in SHE on random graphs. We also explain the implications of
the  lack of  continuous translational symmetry in the corresponding averaged SHE on deterministic Cayley
graph. In Section~\ref{sec.discuss}, we discuss the utility of graphons in the analysis of dynamical systems on graphs
and potential applications of our techniques to related network models.

The approximation result needed for the
analysis of SHE on random graphs is given in the Appendix. It is derived in using the method of \cite{GueVer2016}
based on the concentration inequality for
adjacency matrices of $W$-random graphs.

\section{Discretization}\label{sec.discrete}
\setcounter{equation}{0}

The goal of this  section is to complete the formulation of the discrete model
\eqref{discrete-she} by supplying the details on the families of deterministic and random graphs.
In what follows, we assume that $n$ is even in (\ref{discrete-she}) and denote $N = n/2$. This assumption
is used to simplify computations of the normal forms. We will denote the deterministic graph by $\Gamma^N_W$
and the random graph by $\tl\Gamma^N_W$.

\subsection{The discrete SHE on deterministic graphs}

To define the family of deterministic graphs $(\Gamma^N_W)$, we fix $N\in\N$ and discretize $\T$ as follows.
 Let $h\doteq \frac{1}{2N}$,
      $x_i=ih$,  and 
      \begin{align*}
        Q_i&=\left[ x_i-\frac{h}{2}, x_i+\frac{h}{2}\right), \quad i\in
             [-N+1,  N-1],\\
        Q_N&=\left[ \frac{-1}{2},\frac{-1}{2}+\frac{h}{2}\right)\bigcup \left[ \frac{1}{2}-\frac{h}{2}, \frac{1}{2} \right).
      \end{align*}
      $\Gamma^N_W$ is a weighted graph on $2N$ nodes indexed by integers from $[-N+1, N]$. An edge between nodes $i$ and
      $j$ is supplied with a weight
      $$
      a_{ij}=a_{ji}= (2N)^2 \int_{Q_i}\int_{Q_{j}} W(x,y) dx dy \qquad -N+1\le i<j\le N.
      $$
      In addition, we assume
      $
      a_{ii}=0, \; i\in [-N+1, N].
      $
      
Since $W(x,y) = S(x-y)$ according to \eqref{W-Cay}, we have a Toeplitz matrix $A=(a_{ij})$ with 
      \begin{align}\nonumber
        a_{i+k\, j+k}& =(2N)^2 \int_{Q_{i+k}}\int_{Q_{ j+k}} S(x-y)dxdy\\
        \label{toep}
                     &=(2N)^2 \int_{Q_{i}}\int_{Q_{ j}} S(x-y)dxdy = a_{ij}.
      \end{align}
If $S_k := a_{k0}$, then $a_{ij} = S_{i-j}$ and the graph Laplacian on $\Gamma^N_W$ can be rewritten in the form:
      \begin{align}
        \nonumber
        (L_W^N u)_i &= \frac{1}{2N} \sum_{j=-N+1}^N S_{i-j} \left(u_j -u_i\right)\\
        \nonumber
        &= \frac{1}{2N} \sum_{j=-N+1}^N S_{i-j} u_j  - \left(\frac{1}{2N}\sum_{j=-N+1}^N S_{i-j}\right) u_i\\
        \label{Lap-Ad}
                    &=: \left(A_W^N u\right)_i - d_W^N u_i,
      \end{align}
where $d_W^N$ is a constant independent of $i$. The SHE on the deterministic graph $\Gamma^N_W$ has the following form:
      \begin{equation}\lbl{dshe}
        \dot u = -\left( L^N_W  -\kappa\right)^2 u +\gamma u-u^3,
      \end{equation}
      where $L_W^N$ is defined in \eqref{Lap-Ad}.
      
\subsection{The discrete SHE on random graphs}

The second family of graphs ($\tilde \Gamma_W^N$) is random and corresponds to $\Gamma_W^N, \;N \in \mathbb{N}$.
Denote  the adjacency matrix of $\tl\Gamma_W^N$ by $\tl A^N=(\tl a_{ij})_{-N+1 \leq i, j \leq N}$.
    We postulate that  two distinct nodes of $\tilde \Gamma_W^N$
 $i$ and $j$ are connected with probability $a_{ij}$, i.e.,
\begin{equation}\label{W-1}
\P(\tilde a_{ij}=1)=a_{ij}, \quad \P(\tilde a_{ij}=0)=1-a_{ij}.
\end{equation}
In addition, $\tl a_{ii}=0$ and $\tilde a_{ji}=\tilde a_{ij}$. The SHE on the random graph $\tilde\Gamma_W^N$ is given in the form:
\begin{equation}
\lbl{rshe}
\dot u= -\left( \tilde L^N_W  -\kappa\right)^2 u +\gamma u -u^3, 
\end{equation}
where
\begin{equation}\label{def-L_N}
(\tilde L^N_W u)_i = \frac{1}{N} \sum_{j=1}^N \tilde
a_{ij} \left( u_j-u_i\right) = \left( \tilde{A}_W^N u\right)_i - (\tilde{D}_W^N u)_i,
\end{equation}
where $\tilde{D}_W^N$ is a diagonal matrix.

Graphs defined by  \eqref{W-1} are dense almost surely. In dense graphs  the number of edges scales quadratically with the number of vertices, i.e., under this model $\Gamma^N_W$ has $\mathcal{O}(N^2)$ edges. In this paper, we extend the random graph model to allow for sparse graphs.  To this end, we introduce a nonincreasing sequence $\alpha_N\in (0,1]$ and modify \eqref{W-1} as follows
      \begin{equation}\label{W}
\P(\tilde a_{ij}=\alpha^{-1}_N)=\alpha_N a_{ij}, \quad \P(\tilde a_{ij}=0)=1-\alpha_N a_{ij}.
\end{equation}
As before, $\tl a_{ii}=0$ and $\tilde a_{ji}=\tilde a_{ij}$.

If $\alpha_N\equiv 1$ then \eqref{W} reduces to \eqref{W-1} and we obtain the sequence of dense random graphs
as above. On the other hand, if $\alpha_N\searrow 0$ then the expected number of edges in $\tl\Gamma^N_W$
is $N(N-1)\alpha_N\ll N^2$, which implies that $\tl\Gamma^N_W$ is a sparse graph with probability $1$. By
varying the rate of convergence of $\alpha_N$ to zero, one can control the degree of sparseness of $\tl\Gamma^N_W$.
We need to impose the following technical condition on $(\alpha_N)$:
\begin{equation}\label{cond-alpha}
  1\ge \alpha_N\ge MN^{-1/3},
  \end{equation}
for some $M>0$ dependent of $N$. 

We illustrate the families of graphs, which we just defined, with the following example.
\begin{ex}
	\label{ex.SWgraphon}
	Fix $p\in [0,1]$, $r\in (0,\frac{1}{2})$, and define an even function $S \in L^1(\mathbb{T})$ in the form:
	$$
	S(x) =\left\{
	\begin{array}{ll}
	1-p, & |x| \le r,\\
	p, & r < |x| \le \frac{1}{2}.
	\end{array}
	\right.
	$$
	If $p=0$ $\Gamma^N_W$ is very close to the Cayley graph on $\Z_{2N}$ with the set of generators
	given by $\{ \pm 1, \pm 2, \dots, \pm\lfloor 2Nr\rfloor\}$. For $p\in (0, \frac{1}{2})$,
        $\tl\Gamma^N_W$
	is  a small--world graph \cite{WatStr98}. Figure~\ref{f.pixel} illustrates the example with $W(x,y) = S(x-y)$.
\end{ex}

\begin{figure}[htb!]
	\centering
	\textbf{a}\;\includegraphics[width=.29\textwidth]{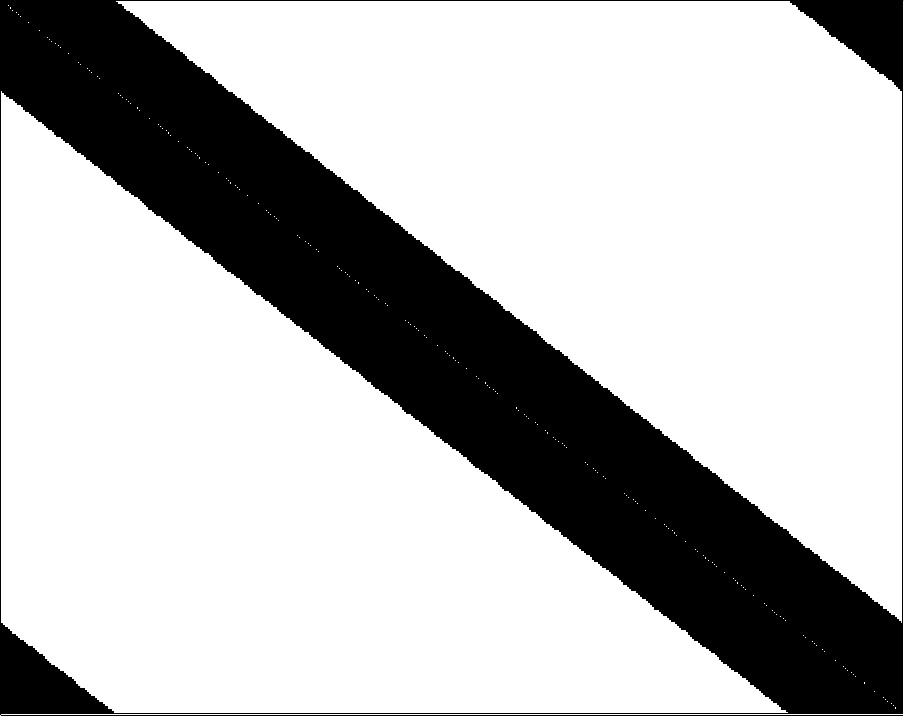}
	\textbf{b}\;\includegraphics[width=.29\textwidth]{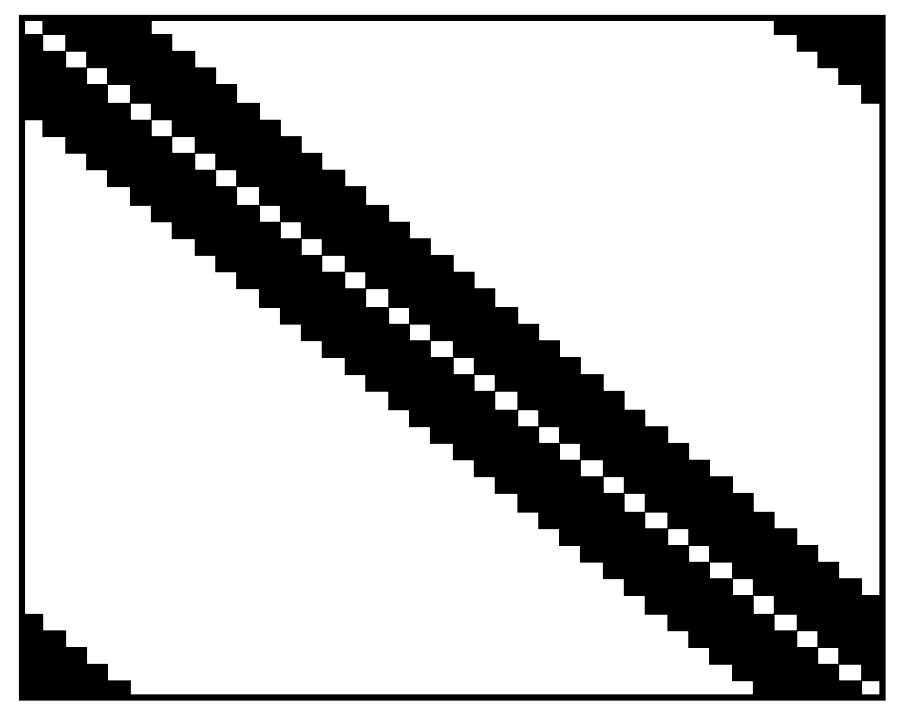}
	\textbf{c}\;\includegraphics[width=.29\textwidth]{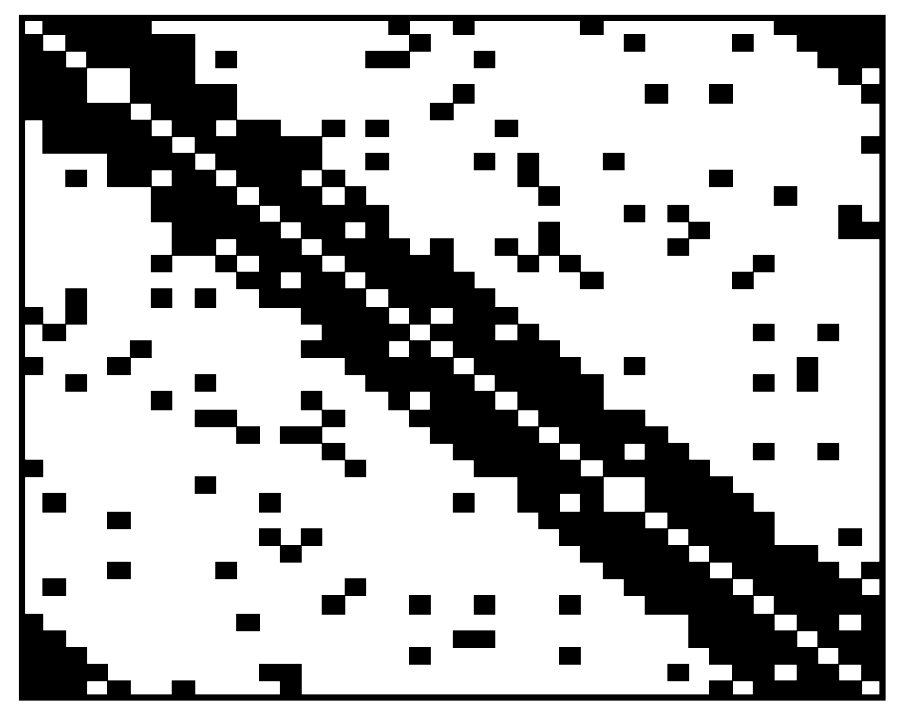}
	\caption{\textbf{(a)}~$W$ takes values $1-p$ and $p$ over the black and
		white regions respectively. \textbf{(b),(c)}~Pixel plots of the adjacency matrix of
		$\Gamma^N_W$  and its random counterpart $\tl\Gamma^N_W$ .}
	\label{f.pixel}
\end{figure}

\begin{rem} 
  The W-random graph interpretation of the small-world network \eqref{W-1} used in this paper 
  was introduced in \cite{Med14b}. A slightly different model was considered in \cite{BraHol23}.
 \end{rem}

\begin{rem}
The deterministic SHE model \eqref{dshe} is a Galerkin approximation of
the continuous SHE model \eqref{nSHE}. On the other hand, it is also 
related to the random model \eqref{rshe} via averaging, because
$\E\tilde a_{ij}=a_{ij}$ by construction. Therefore, one can view
the continuous SHE model \eqref{nSHE} as a continuum limit of either of the discrete models
\eqref{dshe} or \eqref{rshe}. For a related class of nonlocal models,
it is known that the initial-value problems for the 
deterministic and random discrete models approximate that for the
continuum one \cite{Med14a, Med19}. The same techniques
apply to the models at hand and the approximation
results continue to hold for the discrete and continuum SHEs. However,
the analysis in the remainder of this paper does not depend on the
validity of these results.
\end{rem}

\section{The continuum SHE}\label{sec.continuum}
\setcounter{equation}{0}

In this section, we study the nonlocal SHE model (\ref{nSHE}), which serves as a  continuum limit for
the discrete SHE models (\ref{dshe}) and (\ref{rshe}) on deterministic and random graphs $\Gamma^N_W$ and
$\tl\Gamma^N_W$ respectively. Our objective is to obtain a spatially dependent
steady state via a Turing bifurcation of the 
trivial solution. For $W$ in the form \eqref{W-Cay}, the nonlocal SHE on Cayley graphon can be written in the form
\begin{equation}
\label{nonshe}
\partial_t u = -\left(K_S-d_S-\kappa\right)^2 u +\gamma u-u^3,
\end{equation}
where $\kappa$ and $\gamma$ are real parameters and 
\begin{align}\label{def-K_S}
  (K_Sf)(x)& =\int_\T S(x-y)f(y) dy,\\
  \label{def-d_S}
  d_S& =\int_\T S(x) dx.
\end{align}
Throughout this section, we assume that $S$ is a given  integrable function 
on $\mathbb{T}$, i.e., $S \in L^1(\mathbb{T})$. Furthermore, $S$ is assumed to be even almost everywhere. 
By Young's convolution inequality, $K_S$ is a bounded operator 
from $L^p(\mathbb{T})$ to $L^p(\mathbb{T})$ for any $1 \leq p \leq \infty$ with the bound 
\begin{equation}\label{Young}
\| K_S f \|_{L^p(\T)} \leq \| S \|_{L^1(\T)} \| f \|_{L^p(\T)}.
\end{equation}

\begin{rem}
	\label{rem-eigenvalues}
	Since $K_S$ is a bounded operator from $L^2(\mathbb{T})$ to $L^2(\mathbb{T})$ and $\mathbb{T}$ is compact,
        the spectrum $\sigma(L_S)$ is purely discrete and consists of eigenvalues $\{ \lambda_k \}_{k \in \mathbb{Z}}$
        obtained from the Fourier modes:
	$$
	\lambda_k = \int_{\mathbb{T}} S(x) e^{-2\pi i k x} dx, \quad k \in \mathbb{Z}.
	$$ 
	Since  $S$ is an even function, 
	$$
	\sigma(K_S) = \{ \lambda_0, \; \lambda_1 = \lambda_{-1}, \; \lambda_2 = \lambda_{-2}, \; \cdots \}
	$$
	consists of real eigenvalues. However, unless $S$ is even on $\mathbb{T}$, $K_S$ is not a self-adjoint operator in $L^2(\mathbb{T})$
        and the eigenvalues $\{ \lambda_k \}_{k \in \mathbb{Z}}$ only satisfy the reduction $\lambda_k = \bar{\lambda}_{-k}$
        for $k \in \mathbb{N}$ which does not exclude the possibility of complex eigenvalues.
\end{rem}

The following lemma gives the well-posedness result of the time-evolution problem (\ref{nonshe}) in the phase space
$\cX := L^{\infty}(\mathbb{T})$, which is continuously embedded into $L^2(\mathbb{T})$ due to the bound
$\| f \|_{L^2(\T)} \leq \| f \|_{L^{\infty}(\T)}$.

\begin{lem}
	\label{lem-well-posed}
	Assume that $S \in L^1(\mathbb{T})$ and its periodic extension to $\R$ 
        is an even  function. For every $u_0 \in \cX$, there exists a
        unique global solution $u \in C^1([0,\infty),\cX)$ such that $u(0,\cdot) = u_0$.
        The unique solution $u \in C^1([0,\tau_0],\cX)$ is Lipschitz continuous with respect to
        $u_0 \in \cX$ for every finite $\tau_0 > 0$.
\end{lem}

\begin{proof}
	By \eqref{Young} with $p = \infty$,
	$K_S$ is a bounded operator from $\cX$ to $\cX$. In addition, the nonlinear term 
	of (\ref{nonshe}) is closed in $\cX$ due to the bound $\| f^3 \|_{L^{\infty}} \leq \| f \|_{L^{\infty}}^3$. Hence the vector field 
	$$
	A(u) \doteq -(K_S-d_S-\kappa)^2 u + \gamma u - u^3
	$$
	is a $C^1$ map from $\cX$ to $\cX$. By the standard results of the semi-group theory
        \cite[Chapter 3]{Cazenave}, for every $u_0 \in \cX$, there exists a unique local solution $u \in C^1([0,\tau_0],\cX)$ for some $\tau_0 > 0$. Thanks to the repulsive cubic nonlinearity, we have the following bound:
\begin{align*}
  \frac{d}{dt} \| u(t,\cdot) \|_{L^{\infty}(\T)} & \leq \| (K_S-d_S -\kappa)^2 u \|_{L^{\infty}(\T)} +\gamma \| u(t,\cdot) \|_{L^{\infty}(\T)} \\
 & \leq [(\| S \|_{L^1} + |\kappa| + |d_S|)^2 + \gamma] \| u(t,\cdot) \|_{L^{\infty}(\T)}
\end{align*}
which shows that the $L^{\infty}$-norm of the local solution $u(t,\cdot)$ cannot blow up in a finite time $t \in [0,\tau_0]$.
Hence, the local solution $u \in C^1([0,\tau_0],\cX)$ is extended to the infinite time as $u \in C^1([0,\infty),\cX)$.
Lipschitz continuity of the local solution $u \in C^1([0,\tau_0],\cX)$ with respect to $u_0 \in \cX$ for
every finite $\tau_0 > 0$ follows from Gronwall's inequality. 
\end{proof}

\begin{rem}
		\label{rem-space}
                The nonlinear term of (\ref{nonshe}) is not closed in $L^2(\mathbb{T})$. However, it is closed in
                $H^1_{\rm per}(\mathbb{T})$ given by 
	$$
	H^1_{\rm per}(\mathbb{T}) \doteq \left\{ f \in L^2(\mathbb{T}) : \;\; f' \in L^2(\mathbb{T}) \right\},
	$$
	since $H^1_{\rm per}(\mathbb{T})$ is a Banach algebra with respect to pointwise multiplication. Moreover,
        $H^1_{\rm per}(\mathbb{T})$ is continuously embedded into a space of bounded and continuous functions
        satisfying the periodic boundary conditions. Compared to $H^1_{\rm per}(\mathbb{T})$, functions in the
        phase space $\cX = L^{\infty}(\mathbb{T})$ do not have to be continuous or to satisfy the periodic boundary conditions.
        This is more suitable in the context of solutions 
	of the discrete SHE on deterministic and random graphs. 
\end{rem}

\begin{lem}\label{rem.symmetries}
If $u(t,x)$ is a solution of (\ref{nonshe}), so are $u(t,x+h)$, $u(t,-x)$, and $-u(t,x)$.
\end{lem}
\begin{proof}
The continuous SHE \eqref{nonshe} admits the following symmetries:
	\begin{itemize}
		\item the spatial translation $x\mapsto x+h$, $\forall h\in \mathbb{R}$ due to periodic conditions, 
		\item the spatial reflection $x\mapsto -x$ due to even $S$,
		\item the sign reflection $u\mapsto -u$ due to odd nonlinearity,
	\end{itemize}
which can be easily confirmed. The new solutions are generated by the symmetries.
\end{proof}

\begin{rem}
The nonlocal SHE (\ref{nonshe}) has two real parameters $\gamma$ and $\kappa$. 
Parameter $\gamma$ is small and is used to characterize Turing bifurcation 
of a spatially dependent steady state from the zero solution. 
On the other hand, parameter $\kappa$ is the tuning parameter 
defined by the bifurcation condition according to the following definition.
\end{rem}

\begin{assume} \label{ass-bifurcation}
	We fix $k_0 \in \mathbb{N}$, assume that $\lambda_k \neq \lambda_{k_0}$ for every $k \in \mathbb{N}\backslash \{ k_0 \}$, and choose $\kappa := \lambda_{k_0} - d_S = \lambda_{-k_0} - d_S$.
\end{assume}

\begin{rem} \label{rem.separation}
  Since $K_S$ is a compact operator on $L^2(\mathbb{T})$, eigenvalues $\{ \lambda_n \}_{n \in \mathbb{Z}}$
  satisfy $\lambda_n \to 0$ as $|n| \to \infty$. Hence, 
	Assumption \ref{ass-bifurcation} implies that there is $C_0 > 0$ such that 
\begin{equation}
\label{eigenvalue-bounded-away}
|\lambda_{k} - \lambda_{k_0}| \geq C_0 \;\; \mbox{\rm for all } k \in \mathbb{N}, \;\; k \neq k_0.
\end{equation}
\end{rem}

The following theorem presents the main result of this section.
Although it is an exercise from \cite[Section 2.4.3]{HaragusIooss}, we write 
the computational details explicitly, since they are useful for analysis 
of Turing bifurcation in the discrete SHE.

\begin{thm}\label{thm.Turing}
  Under Assumption \ref{ass-bifurcation}, there exists $\gamma_0 > 0$ and $C_0 > 0$
  such that for every $\gamma \in (0,\gamma_0)$ there exists a  non-trivial time-independent solution
  $u_{\gamma}(\cdot + \delta)$ in $\cX$ of the nonlocal SHE model (\ref{nonshe}), where $u_{\gamma}$ is an even function satisfying
	\begin{equation} 
\label{Stokes}
		\sup_{x \in \mathbb{T}} \left| u_{\gamma}(x) - \frac{2 \sqrt{\gamma}}{\sqrt{3}}  
		\cos(2\pi k_0 x) \right| \leq C_0 \sqrt{\gamma^3},
	\end{equation}
and $\delta \in \mathbb{T}$ is an arbitrary translational parameter. The 
orbit of time-independent solutions $\{ u_{\gamma}(\cdot + \delta) \}_{\delta \in \mathbb{T}}$ is asymptotically stable in the time evolution of the nonlocal SHE in $\cX$.
	\label{theorem-SH-local}
\end{thm}

\begin{proof}
In order to use the center manifold theorem and to derive slow dynamics 
along the center manifold, we use the Fourier series 
$$
u(t,x) = \sum_{k \in \Z} a_k(t) e^{2 \pi \iu k x}
$$
and obtain the evolution problem in the form 
\begin{equation}\label{she-Fourier}
\dot{a}_k = -(\lambda_k - \lambda_{k_0})^2 a_k + \gamma a_k - \sum_{k_1,k_2 \in \Z} a_{k_1} a_{k_2} a_{k-k_1-k_2},
\end{equation}
where we have used that $\kappa = \lambda_{k_0} - d_S$. 
By standard results from Fourier analysis, $u(t,\cdot) \in \cX = L^{\infty}(\mathbb{T})$ if $\{a_k(t) \}_{k \in \mathbb{Z}} \in \hat{\cX} = \ell^{1}(\mathbb{Z})$
for every $t \geq 0$. The nonlinear term of 
(\ref{she-Fourier}) is closed since $\ell^1(\Z)$ is a Banach algebra with respect to the convolution sum.

Since $u$ is real, we have $a_{-k} = \bar{a}_k$ for all $k \in \mathbb{Z}$.
Due to the three symmetries identified in Lemma~\ref{rem.symmetries}, the vector field in (\ref{she-Fourier}) is
equivariant under the transformation 
\begin{equation}\label{transform-a}
a_k \to a_k e^{2\pi \iu k h}, \quad k \in \Z, \quad  h \in \R, 
\end{equation}
and under the transformations: $a_k \to a_{-k}$ and $a_k \to -a_k$.  Consequently, the system
(\ref{she-Fourier}) is closed on the subspaces 
\begin{equation}\label{closed-real}
\hat{\cX}_{\rm sym} := \{ \{ a_k \}_{k \in \Z} \in \ell^{1}(\Z,\R) : \;\; 
a_{-k} = a_k \; \} 
\end{equation}
and
\begin{equation}\label{closed-bif}
\hat{\cX}_{\rm bif} := \{ \{ a_k \}_{k \in \Z} \in \ell^{1}(\Z,\mathbb{C}) : \;\; 
a_{-k} = a_k = 0, \; \mbox{\rm for } k \neq m k_0, \; m \in \{1,3,5,\cdots\} \},
\end{equation}
where $k_0 \in \mathbb{N}$ is defined in Assumption \ref{ass-bifurcation}.

By Assumption \ref{ass-bifurcation}, the rest of eigenvalues
$\{ \lambda_k \}_{k \in \mathbb{Z} \backslash \{ \pm k_0\}}$ are bounded away from $\lambda_{k_0} = \lambda_{-k_0}$ with the bound (\ref{eigenvalue-bounded-away}).
By the center manifold theorem \cite[Theorem 2.9]{HaragusIooss}, there exists a
center manifold in $\hat{\cX}_{\rm bif}$ spanned by $ A:=a_{k_0} \in \mathbb{C}$ and $ \bar{A}:=a_{-k_0}$.
Since the system is closed on (\ref{closed-bif}), the center manifold can be expressed as graphs of functions: 
$$
a_{m k_0} = \Psi_m(A,\bar{A}), \quad a_{-mk_0} = \bar{\Psi}_m(A,\bar{A}), \quad m \in \{3,5,\cdots \}.
$$ 
The dynamics on the center manifold can be expressed by the amplitude 
equations 
$$
\dot{A} = F_1(A,\bar{A}), \quad \dot{\bar{A}} = \overline{F_1(A,\bar{A})},
$$
where $F_1$ is a $C^{\infty}$ function in $A$ and $\bar{A}$ with $\gamma$-dependent coefficients which commutes with the symmetries of (\ref{she-Fourier}). Due to the equivariance (\ref{transform-a}), the amplitude equations can be transformed to the normal form:
\begin{equation}\label{eqn-A}
\dot{A} = A P_1(|A|^2),
\end{equation}
where $P_1$ is an analytic function in $|A|^2$ with $\gamma$-dependent coefficients. Moreover, due to the symmetry with respect to the transformation 
$a_k \to a_{-k}$, $P_1$ has real-valued coefficients. Similarly, 
we express functions $\Psi_m$ in the form:
\begin{equation}\label{eqn-Psi}
\Psi_m(A,\bar{A}) = A^m P_m(|A|^2), \quad m \in \{3,5,\cdots \},
\end{equation}
where $P_m$ is a $C^{\infty}$ function in $|A|^2$ with $\gamma$-dependent real-valued coefficients.

Due to the cubic nonlinearity in (\ref{she-Fourier}), we obtain 
\begin{equation}\label{trunc-A}
\dot{A}= \left[ \gamma  - 3 |A|^2 + \mathcal{O}(|A|^4) \right] A,
\end{equation}
where the remainder terms of the order of $\mathcal{O}(|A|^4)$ is defined by $\Psi_3$ 
since $\Psi_m$ with $m \geq 5$ give a higher-order contribution of $\mathcal{O}(|A|^6)$ to the normal form (\ref{trunc-A}). 
It follows from (\ref{trunc-A}) 
that there exists a time-independent solution of the form 
\be \label{steady}
A_{\gamma,\delta} := \frac{\sqrt{\gamma}}{\sqrt{3}} \left[ 1  + \mathcal{O}(\gamma) \right] e^{2 \pi k_0 \iu \delta},
\end{equation}
where $\delta$ is an arbitrary parameter. This time-independent solution 
(\ref{steady}) yields a non-trivial time-independent solution of the nonlocal SHE (\ref{nonshe}) in $\cX$
satisfying the expansion in (\ref{Stokes}). Since the system (\ref{she-Fourier}) is closed on (\ref{closed-real}),
if $\delta = 0$, then $a_{-k_0} = a_{k_0} = A_{\gamma,\delta = 0}$ is real and so are $a_{\pm m k_0}$ for every
$m \in \{ 3,5,\cdots\}$. This yields the even function $u_{\gamma} \in \cX$. Due to the translational symmetry (\ref{transform-a}), the parameter $\delta$ is equivalent to the translation of the solution $u_{\gamma}(\cdot + \delta)$. 

To determine stability of the time-independent solution 
$u_{\gamma}(\cdot + \delta)$ for every given $\delta \in \mathbb{R}$, 
we note that all eigenvalues  in the spectrum of $-(K_S - d_S - \kappa)^2$ are located in the
left-half plane of the complex plane with the exception of the double zero eigenvalue. Hence,
there is no unstable manifold of the system (\ref{she-Fourier}). The time-independent solution (\ref{steady})
is orbitally asymptotically stable in the time evolution of the reduced equation (\ref{eqn-A}).
By the standard decomposition near the orbit 
$\{ u_{\gamma}(\cdot + \delta)\}_{\delta \in \mathbb{T}}$, 
a perturbation of the initial solution $u_{\gamma}(\cdot + \delta_0)$ defines 
a time-dependent solution which approaches as $t \to +\infty$ 
exponentially fast to the final solution  $u_{\gamma}(\cdot + \delta_{\infty})$ 
with $\delta_{\infty}$ being close to $\delta_0$. Hence, the orbit 
$\{ u_{\gamma}(\cdot + \delta)\}_{\delta \in \mathbb{T}}$  is asymptotically stable in
the time evolution of the nonlocal SHE in $\cX$.
\end{proof}

\begin{rem}
The leading-order approximation for $\Psi_3$ in (\ref{eqn-Psi}) can be obtained 
from (\ref{she-Fourier}) and (\ref{trunc-A}). Substitution yields
	\begin{align*}
	& [\gamma - 3 |A|^2 + \mathcal{O}(|A|^4)] \left[  3 P_3(|A|^2) 
	+ |A|^2 P_3'(|A|^2) \right] \\
	& \quad =
	\left[ \gamma - (\lambda_{3k_0} - \lambda_{k_0})^2 \right] P_3(|A|^2) 
	- 1 + \mathcal{O}(|A|^2),
	\end{align*}
which is solved by expanding $P_3(|A|^2)$ in $|A|^2$ with the leading order:
	\begin{equation}
	\label{trunc-Psi}
	P_3(|A|^2) = \frac{-1}{(\lambda_{3k_0} - \lambda_{k_0})^2 + 2 \gamma} + \mathcal{O}(|A|^2),
	\end{equation}
	where $\lambda_{3k_0}\neq \lambda_{k_0}$ by the assumption.
	Similarly, one can find the leading-order expansions of $\Psi_m$ in (\ref{eqn-Psi}) for $m \geq 5$.
	\end{rem}

Related to the time-independent solution $u_{\gamma}(\cdot + \delta)$ of the nonlocal SHE in Theorem \ref{theorem-SH-local}, we can introduce the linearized operator 
in the form 
\begin{equation}
\label{lin-operator}
\mathcal{L}_{\gamma} := -(K_S-\lambda_{k_0})^2 +\gamma - 3 u_{\gamma}^2,
\end{equation}
where we have used that $\kappa = \lambda_{k_0} - d_S$ according to 
Assumption \ref{ass-bifurcation}. Since $u_{\gamma}$ is a bounded function on $\mathbb{T}$,
$\mathcal{L}_{\gamma}$ is a bounded operator from $L^2(\mathbb{T})$ to $L^2(\mathbb{T})$ for all sufficiently small 
$\gamma$ for which $u_{\gamma}$ is defined. Hence, the spectrum of $\mathcal{L}_{\gamma}$
is purely discrete and consists of real eigenvalues, see Remark \ref{rem-eigenvalues}.
The following lemma uses the smallness of $\gamma$ and gives a precise information on the
location of these eigenvalues. 

\begin{lem}
  Let $u_{\gamma} \in \cX$ be defined by Theorem \ref{theorem-SH-local} for $\gamma \in (0,\gamma_0)$ and
  assume that $u_{\gamma}' \in L^2(\mathbb{T})$. The spectrum of $\mathcal{L}_{\gamma}$ given by (\ref{lin-operator}) in $L^2(\mathbb{T})$ consists of eigenvalues ordered as 
$\{ \Lambda_k(\gamma) \}_{k \in \mathbb{N}}$ such that 
$\Lambda_1(\gamma) = 0$ and $\Lambda_k(\gamma) < 0$, $k \geq 2$ satisfy
$$
|\Lambda_2(\gamma)| \leq C_1 \gamma, \quad |\Lambda_k(\gamma)| \geq C_2, \;\; k \geq 3
$$
where $C_1$, $C_2$ are $\gamma$-independent positive constants.
	\label{lemma-SH-nonlocal}
\end{lem}

\begin{proof}
  The existence of $\Lambda_1(\gamma) = 0$ follows from the translational invariance of the nonlocal
  SHE given by (\ref{nonshe}) since $\delta \in \mathbb{R}$ is a free parameter of the steady state
  $u_{\gamma}(\cdot + \gamma) \in \cX$ in Theorem \ref{theorem-SH-local}.
  Since $u_{\gamma}' \in L^2(\mathbb{T})$ is assumed, we obtain by direct differentiation that 
$$
\mathcal{L}_{\gamma} u_{\gamma}'(\cdot + \delta) = 0,
$$ 
so that $\Lambda_1(\gamma) = 0$ for $\gamma \in (0,\gamma_0)$.

The rest of the spectrum of $\mathcal{L}_{\gamma}$ in $L^2(\mathbb{T})$ follows from the perturbation theory
for self-adjoint operators with purely discrete spectrum which implies continuity of
eigenvalues with respect to small parameter $\gamma \in (0,\gamma_0)$. The eigenvalue $\Lambda_2(\gamma)$
coincides with the linearization of the slow motion at the center manifold 
given by (\ref{trunc-A}):
\begin{align*}
\Lambda_2(\gamma) &= \gamma  - 9 |A_{\gamma,\delta}|^2 + \mathcal{O}(|A_{\gamma}|^4) \\
&= - 2 \gamma + \mathcal{O}(\gamma^2),
\end{align*}
which is strictly negative with the bound $|\lambda_2(\gamma)| \leq C_1 \gamma$ with some $\gamma$-independent constant $C_1 > 0$ for $\gamma \in (0,\gamma_0)$.
Eigenvalues $\Lambda_k(\gamma)$ for $k \geq 3$ are $\gamma$-close to 
$-(\lambda_k-\lambda_{k_0})^2$ for $k \in \mathbb{Z}\backslash \{k_0,-k_0\}$ which are 
strictly negative and bounded away from $0$ by Assumption \ref{ass-bifurcation}.
Hence $|\Lambda_k(\gamma)| \geq C_2$ for all $k \geq 2$ with some
$\gamma$-independent constant $C_2 > 0$ for $\gamma \in (0,\gamma_0)$.
\end{proof}

\begin{rem}
  Due to the presence of zero eigenvalue $\Lambda_1(\gamma) = 0$, the operator $\mathcal{L}_{\gamma}$ is not invertible.
  This creates difficulties in the persistence argument when the limiting nonlocal model (\ref{nonshe}) is replaced 
	by the discrete SHE models on the deterministic or random graphs, which destroy the continuous translational symmetry.
\end{rem}

\section{The discrete SHE on Cayley graphs}\label{sec.Cayley}
\setcounter{equation}{0}

Here we study the discrete SHE \eqref{dshe} on the deterministic Cayley graph $\Gamma^N_W$. By using \eqref{Lap-Ad}, we rewrite the discrete model in the form:
\begin{equation}
\label{pertshe}
\dot u_j = -[(A^N_W-d_W^N-\kappa)^2 u]_j + \gamma u_j - u_j^3, \quad  j\in \Z_N := \Z/(2N\Z),
\end{equation}
where $N$ is integer, $\kappa$ and $\gamma$ are real parameters, and the linear operator $A^N_W : \Z_N \to \Z_N$ is given by the convolution sum
\begin{equation}
\label{def-LN}
[A^N_W u]_j = \frac{1}{2N} \sum_{l \in \Z_N} S_{j-l} u_l, \quad j \in \Z_N,
\end{equation}
where $\{ S_j \}_{j \in \Z_N}$ satisfies $S_{-j} = S_j$ for all $j \in \Z_N$. 
System of differential equations \eqref{pertshe} can be viewed as an evolution equation on $\Z_N$. The classical solutions are interpreted as 
elements of $C^1(\R,\R^{\Z_N})$, the space of continuously differentiable vector functions of $t \in \R$.

\begin{lem}\label{obs.discrete}
If $\{ u_j(t) \}_{j \in \Z_N}$ is a solution of the discrete SHE (\ref{pertshe}), so are 
\begin{equation}
\label{new-sol}
\{ u_{j+m}(t) \}_{j \in \mathbb{Z}_N}, \quad  
\{ u_{-j}(t) \}_{j \in \mathbb{Z}_N}, \quad 
\{ u_{m-j}(t) \}_{j \in \mathbb{Z}_N}, \;\; \mbox{\rm and } \;\;  
\{ -u_{j}(t) \}_{j \in \mathbb{Z}_N}.
\end{equation}
\end{lem}
\begin{proof}
  Similarly to the continuous SHE (\ref{nonshe}), the discrete SHE  (\ref{pertshe}) admits the following three symmetries:
\begin{itemize}
	\item the discrete spatial translation $j \mapsto j + m$, $\forall m \in \Z_N$ due to periodic conditions,
	\item the spatial reflection $j \mapsto -j$ due to even $\{ S_j \}_{j \in \Z_N}$, 
	\item the sign reflection $u \mapsto -u$ due to odd nonlinearity,
\end{itemize}
which can be easily confirmed. The new solutions (\ref{new-sol}) are generated from $\{ u_j(t) \}_{j \in \Z_N}$ by symmetries. 
\end{proof}

Our objective is to obtain a spatially dependent steady state of the discrete SHE (\ref{pertshe}) via a
Turing bifurcation of the zero solution. By using the discrete Fourier transform, one can obtain eigenvalues of $A^N_W$ in the
form $\{ \lambda^N_k \}_{k \in \Z_N}$ with 
\begin{equation}
\label{eig-discrete}
\lambda^N_k = \frac{1}{2N}\sum_{j \in \Z_N} S_j e^{-\frac{i \pi k j}{N}}, \quad k \in \Z_N.
\end{equation}
Since $S_{-j} = S_j$ for all $1 \leq j \leq N-1$, then $\lambda^N_k = \lambda^N_{-k}$ for $1 \leq k \leq N-1$.
We again use parameter $\gamma$ in (\ref{pertshe}) to characterize Turing bifurcation and parameter $\kappa$ to
satisfy the bifurcation condition. 

We will locate a  spatially dependent steady state bifurcating from the zero solution by adapting the proof of Theorem
\ref{theorem-SH-local} with the discrete Fourier transform replacing Fourier series. One key new feature of the center manifold analysis is that the translational parameter $\delta \in \mathbb{T}$ which was arbitrary in Theorem \ref{theorem-SH-local} takes exactly $4 N$ admissible values under some non-degeneracy conditions. In comparison
with Assumption \ref{ass-bifurcation}, we set $k_0 = 1$ for the bifurcating mode in
order to simplify computations of the normal form. The following theorem presents the main
result of this section.

\begin{thm}\label{thm.discrete-Turing}
  Assume that $\lambda^N_k \neq \lambda^N_{1}$ for $k \neq \pm 1$ and choose
  $\kappa := \lambda^N_{1} - d^N_W= \lambda^N_{-1} - d^N_W$. If $r_N \neq 0$ in (\ref{non-degeneracy-1}),
  then there exists $\gamma_0 > 0$ and $C_0 > 0$ such that for every $\gamma \in (0,\gamma_0)$ and every integer $N \geq 3$ there exist two non-trivial time-independent solutions
  $u^N_{\gamma}, v^N_{\gamma} \in \R^{\Z_N}$ of the discrete SHE (\ref{pertshe}), where $u^N_{\gamma}$ is symmetric about $j = 0$ and satisfies
	\begin{equation} 
	\label{Stokes-discrete}
	\sup_{j \in \Z_N} \left| u_j - \frac{2}{\sqrt{3}}  \sqrt{\gamma}
	\cos\left(\frac{\pi j}{N}\right) \right| \leq C_0 \sqrt{\gamma^3}.
	\end{equation}
	and $v^G_{\gamma}$ is symmetric about the mid-point between $j = 0$ and $j = 1$ and satisfies
	\begin{equation} 
	\label{Stokes-discrete-2}
	\sup_{j \in \Z_N} \left| u_j - \frac{2}{\sqrt{3}}  \sqrt{\gamma}
          \cos\left(\frac{\pi j}{N} - \frac{\pi}{2N} \right) \right| \leq C_0 \sqrt{\gamma^3}.
	\end{equation}
One of the two solutions is asymptotically stable in the time evolution of the discrete SHE in $C^1(\R,\R^{\Z_N})$ and the other one is unstable. These solutions generate $(2N)$ asymptotically stable and $(2N)$ unstable solutions on $\Z_N$ via the discrete group of spatial translations.
	\label{theorem-SH-discrete}
\end{thm}

\begin{proof}
  We use the discrete Fourier transform
  \begin{equation}\label{discreteF}
    u_j = \sum_{k \in \Z_N} a_k(t)  e^{\frac{\iu \pi k j}{N}}, \quad j \in \Z_N,
    \end{equation}
	with real $a_0$, $a_{N}$, and possibly complex $a_{-k} = \bar{a}_k$ for $1 \leq k \leq N-1$.
        The Fourier amplitudes are extended periodically 
	with the period $2N$ as the sequence $\{ a_k \}_{k\in \Z_N}$. The discrete SHE (\ref{pertshe}) transforms to the
        evolution problem in the form 
	\begin{align}
\notag
	\dot{a}_k &= -(\lambda^N_k - \lambda^N_1)^2 a_k + \gamma a_k - \sum_{-N \leq k_1,k_2,k-k_1-k_2 \leq N-1} a_{k_1} a_{k_2} a_{k-k_1-k_2} \\
	\notag
	& \qquad \qquad - \sum_{-N \leq k_1,k_2,k+2N-k_1-k_2 \leq N-1} a_{k_1} a_{k_2} a_{k+2N-k_1-k_2} \\
	& \qquad \qquad - \sum_{-N \leq k_1,k_2,k-2N-k_1-k_2 \leq N-1} a_{k_1} a_{k_2} a_{k-2N-k_1-k_2},
		\label{she-discrete}
	\end{align}
where the summation terms are adjusted due to the $(2N)$-periodicity of the Fourier modes in the Fourier space and we have used that $\kappa = \lambda^N_1 - d_W^N$. Symmetries in Lemma~\ref{obs.discrete} imply that the system (\ref{she-discrete}) is equivariant
under the transformation 
	\begin{equation}\label{transform-a-discrete}
	a_k \to a_k e^{\frac{\iu \pi k m}{N}}, \quad k,m\in\Z_N
	\end{equation}
        and under the transformations: $a_k \to a_{-k}$ and $a_k \to -a_k$. In particular, the system (\ref{transform-a-discrete})
        is closed on the subspaces
	\begin{equation}
	\label{closed-real-discrete}
\{ a \in \R^{\Z_N} : \;\; 
	a_{-k} = a_k, \;\; 1 \leq k \leq N-1  \}
	\end{equation}
	and
	\begin{equation}
	\label{closed-bif-discrete}
\{ a \in \mathbb{C}^{\Z_N} : \;\; 
a_{-k} = a_k = 0, \; \mbox{\rm for } k \neq m, \; m \in \{1,3,5,\dots, M \} \},
	\end{equation}
	where $M = N$ if $N$ is odd or $M = N-1$ if $N$ is even.
		
We apply the center manifold theorem \cite[Theorem 2.9]{HaragusIooss} under the assumption that $\lambda^N_k \neq \lambda^N_{1}$ for $k \neq \pm 1$. Similar to the proof of Theorem \ref{theorem-SH-local}, there exists a center manifold of the system
(\ref{she-discrete}) spanned by $a_{1} \equiv A \in \mathbb{C}$ and $a_{-1} = \bar{A}$. Since the system is closed on (\ref{closed-bif-discrete}), the center manifold can be expressed as graphs of functions: 
\begin{equation}
\label{center-manifold}
a_{m} = \Psi_m(A,\bar{A}), \quad 
a_{-m} = \bar{\Psi}_m(A,\bar{A}), \quad m \in \{3,5,\dots, M \}.
\end{equation}
The slow dynamics on the center manifold can be expressed by the amplitude equations 
$$
\dot{A}= F_1(A,\bar{A}), \quad \dot{\bar{A}} = \overline{F_1(A,\bar{A})},
$$
where $F_1$ is a $C^{\infty}$ function in $A$ and $\bar{A}$ with $\gamma$-dependent
coefficients which commutes with the symmetries of (\ref{she-discrete}).
Compared to the normal form for the amplitude equations in (\ref{eqn-A}) and (\ref{eqn-Psi})
the discrete Fourier modes are $(2N)$-periodic so that 
$$
\left[  A e^{\frac{\iu \pi k j}{N}} \right]^{2N+1} =  A^{2N+1} e^{\frac{\iu \pi k j}{N}} \quad \mbox{\rm and} \quad \left[  \bar{A} e^{\frac{-\iu \pi k j}{N}} \right]^{2N-1} =  \bar{A}^{2N-1} e^{\frac{\iu \pi k j}{N}}.
$$
Following the classification of normal forms under the symmetry (\ref{transform-a-discrete}) in \cite{Chossat}, we transform 
the amplitude equations to the normal form: 
\begin{equation}
\label{eqn-A-discrete}
\dot{A} = A Q_1(|A|^2,A^{2N},\bar{A}^{2N}) + \bar{A}^{2N-1} R_1(|A|^2,A^{2N},\bar{A}^{2N}),
\end{equation}
where $Q_1$ and $R_1$ are $C^{\infty}$ functions in $|A|^2$, $A^{2N}$, and $\bar{A}^{2N}$ with $\gamma$-dependent coefficients.
Due to the symmetry with respect to the transformation $a_k \to a_{-k}$, $Q_1$ and $R_1$ have real-valued coefficients.
Similarly, we express functions $\Psi_m$ in the form:
	\begin{align}
	\Psi_m(A,\bar{A}) &= A^m Q_m(|A|^2,A^{2N},\bar{A}^{2N}) 
	+ \bar{A}^{2N-m} R_m(|A|^2,A^{2N},\bar{A}^{2N}),
		\label{eqn-Psi-discrete}
	\end{align}
        where $Q_m$ and $R_m$ are $C^{\infty}$ functions in $|A|^2$, $A^{2N}$, and $\bar{A}^{2N}$ with $\gamma$-dependent real-valued coefficients. 

We assume the following non-degeneracy condition: 
	\begin{equation}
	\label{non-degeneracy-1}
	r_N := R_1(0,0,0) \neq 0,
	\end{equation}
where $r_N$ is a $\gamma$-dependent real-valued coefficient. 
Similarly to (\ref{trunc-A}), we obtain from the cubic nonlinearity in (\ref{she-discrete}) that 
$$
Q_1(|A|^2,A^{2N},\bar{A}^{2N}) = \gamma - 3 |A|^2 + \mathcal{O}(|A|^4),
$$
if $N \geq 3$. By using the polar form $A = \rho e^{\iu \theta}$, we write
\begin{equation*}
\left\{ \begin{array}{l} 
\dot{\rho} = \rho [\gamma - 3 \rho^2 + \mathcal{O}(\rho^4)] + \cos(2N\theta) \rho^{2N-1} [ r_N + \mathcal{O}(\rho^2) ], \\
\dot{\theta} = -\sin(2N\theta) \rho^{2N-2} [ r_N + \mathcal{O}(\rho^2)].
\end{array} \right.
\end{equation*}
If $r_N \neq 0$, there exist two distinct time-independent solutions for $\theta = 0$ and $\theta = \frac{\pi}{2N}$ on
interval $[0,\frac{\pi}{N})$. If $N \geq 3$, both solutions can still be expressed at the leading order in the form:
\begin{equation}
\label{steady-discrete}
A_{\gamma,\delta} = \frac{\sqrt{\gamma}}{\sqrt{3}} \left[ 1 + \mathcal{O}(\gamma) \right] e^{\frac{i \pi \delta}{N}},
\end{equation}
with either $\delta = 0$ (which corresponds to $\theta = 0$) or $\delta = \frac{1}{2}$ (which corresponds to
$\theta = \frac{\pi}{2N}$). For $\delta = 0$, we get the real solution $\{ a_k \}_{k \in \Z_N}$ on the subspace (\ref{closed-real-discrete}) that corresponds to the solution $u_{\gamma}^N$ of the discrete SHE (\ref{pertshe}) which is symmetric about $j = 0$ and satisfies (\ref{Stokes-discrete}). For $\delta = \frac{1}{2}$, we get a 
complex solution  $\{ a_k \}_{k \in \Z_N}$ that corresponds to the real solution $v_{\gamma}^N$ of the discrete SHE
(\ref{pertshe}) which satisfies  (\ref{Stokes-discrete-2}). The leading order and hence the solution $v_{\gamma}^N$ is
symmetric about the mid-point between $j = 0$ and $j = 1$ due to the symmetry 
with respect to the transformation $u_j \to u_{1-j}$  in Lemma \ref{obs.discrete}.

To obtain the stability conclusion for the time-independent solutions, we observe that all eigenvalues  in the spectrum of
$-(A^N_W - \lambda^N_1)^2$ are located in the left-half plane of the complex plane with the exception of the double zero eigenvalue. 
If $r_N \neq 0$, then linearization of the time-dependent 
equation (\ref{eqn-A-discrete}) at the leading-order solution (\ref{steady-discrete}) for $N \geq 3$ yields for the perturbations $(\rho',\theta')$:
\begin{equation}
\left\{ \begin{array}{l} 
\dot{\rho}' = -2 \gamma \rho' + \mathcal{O}(\gamma^2), \\
\dot{\theta}' = \mp \left(\frac{\gamma}{3}\right)^{N-1}  \left[ r_N + \mathcal{O}(\gamma) \right] \theta', \end{array} \right.
\label{lin-eqs}
\end{equation}
where the upper sign corresponds to the solution $u_{\gamma}^N$ and the lower sign corresponds to the solution $v_{\gamma}^N$.
Dynamics of (\ref{lin-eqs}) in $\rho'$ is asymptotically stable, whereas dynamics of (\ref{lin-eqs}) in $\theta'$ is either asymptotically stable or unstable. This yields the conclusion that one of the two solutions $u_{\gamma}^N$ and $v_{\gamma}^N$ is asymptotically stable in the time evolution of the discrete SHE in $\mathbb{R}^{2N}_{\rm per}$ and the other one is unstable.

The two solutions $u_{\gamma}^N$ and $v_{\gamma}^N$ generate $(4N)$ solutions by using the discrete group of translations $u_j \mapsto u_{j \pm m}$ for every $m \in \Z_N$ by Lemma \ref{obs.discrete}. The stability of the translated solutions coincide with the stability of $u_{\gamma}^N$ and $v_{\gamma}^N$.
\end{proof}

\begin{rem}
	Since 
	$$
	\cos\left( \frac{\pi j}{N} - \frac{\pi}{N} \right) = -\cos\left(\frac{\pi j}{N} \right),
	$$
	$N$ of $(2N)$ states obtained from either (\ref{Stokes-discrete}) or (\ref{Stokes-discrete-2}) with the discrete group of translations are the sign reflections of the other $N$ states, in accordance with the sign reflection symmetry in Lemma \ref{obs.discrete}.
\end{rem}

\begin{rem} 
	The graphons defined by the graphs $\Gamma_W^N$ used the formulation of the discrete
  SHE (\ref{pertshe}) converge to $W$, the kernel used in the continuous SHE (\ref{nonshe}), in the cut-norm as
  $N\to\infty$. Thus, $\lambda_k^N \rightarrow \lambda_k $ for every fixed $k \in \mathbb{Z}$ \cite{Sze-2011}.
  From this we conclude that Assumption \ref{ass-bifurcation} with $k_0 = 1$ implies the corresponding assumption
  of Theorem \ref{theorem-SH-discrete}, i.e.,  $\lambda^N_k \neq \lambda^N_{1}$, $k \neq \pm 1$ for sufficiently
  large $N$.
\end{rem}

\begin{rem}
	If the non-degeneracy condition (\ref{non-degeneracy-1}) is not satisfied, one needs to expand functions $Q_1$ and $R_1$ in (\ref{eqn-A-discrete}) to the higher orders and to obtain the higher-order non-degeneracy conditions. If it happens that 
	$R_1 \equiv 0$ and $Q_1(|A|^2,A^{2N},\bar{A}^{2N}) = P_1(|A|^2)$, then the time-independent solution (\ref{steady-discrete}) exists with arbitrary $\delta \in \mathbb{R} / \mathbb{Z}$. In this degenerate case, there exists a family of non-trivial time-independent solutions $u_{\gamma,\delta}^N \in \R^{\Z_N}$ of the discrete SHE (\ref{pertshe}) with arbitrary parameter $\delta \in [0,1]$, where $u_{\gamma,\delta}$ satisfies
	\begin{equation} 
	\label{Stokes-discrete-3}
	\sup_{-N \leq j \leq N} \left| u_j - \frac{2}{\sqrt{3}}  \sqrt{\gamma}
	\cos\left(\frac{\pi j}{N} - \frac{\pi \delta}{2N} \right) \right| \leq C_0 \sqrt{\gamma^3}.
	\end{equation}
	The orbit of time-independent solutions $\{ u_{\gamma,\delta}^N \}_{\delta \in \mathbb{R} / \mathbb{Z}}$ is asymptotically stable in the time evolution of the discrete SHE in $C^1(\R,\R^{\Z_N})$. Although such degenerate cases may exist in other discrete models, see, e.g., \cite{HPS11}, the explicit computations for the particular discrete SHE model (\ref{pertshe}) show that $r_N \neq 0$ for every $N \geq 3$.
\end{rem}

\begin{rem}
	Cases $N = 1$ and $N = 2$ are exceptional. In both cases, the discrete Fourier transform in the subspace (\ref{closed-bif-discrete}) is given by the sum of two terms
\begin{equation}
\label{two-terms}
	u_j = A(t) e^{\frac{\iu \pi j}{N}} + \bar{A}(t) e^{-\frac{\iu \pi j}{N}},
\end{equation}
If $N = 1$, then $A$ in (\ref{two-terms}) is real and satisfies 
	$$
	\dot{A} = \gamma A - 4 A^3,
	$$
	from which the stable nontrivial solutions at $A = \pm \sqrt{\gamma}/2$ exist in addition to the unstable solution $A = 0$ for $\gamma > 0$.
If $N = 2$, then $A$ in (\ref{two-terms}) is complex and satisfies 
$$
\dot{A} = \gamma A - 3 |A|^2 A - \bar{A}^3.
$$
Using the polar form $A = \rho e^{\iu \theta}$, this equation is reduced 
to the system 
\begin{equation*}
	\left\{ \begin{array}{l} 
\dot{\rho} = \gamma \rho - 3 \rho^3 - \rho^3 \cos(4\theta), \\
\dot{\theta} = \rho^2 \sin(4 \theta), \end{array} \right.
\end{equation*}
from which the two nontrivial time-independent solutions are given by 
$$
(\rho,\theta) = \left(\frac{\sqrt{\gamma}}{2},0\right) \quad \mbox{\rm and} \quad 
(\rho,\theta) = \left(\frac{\sqrt{\gamma}}{\sqrt{2}},\frac{\pi}{4}\right).
$$
The linearization shows that the first solution is linearly unstable and the 
second solution is asymptotically stable. 
\end{rem}
	
The following two examples give details of computations of the center manifold reductions for $N =3$ and $N = 4$, for which the discrete Fourier transform in the subspace (\ref{closed-bif-discrete}) is given by the sum 
\begin{equation}
\label{four-terms}
u_j = A(t) e^{\frac{\iu \pi j}{N}} + \bar{A}(t) e^{-\frac{\iu \pi j}{N}} + B(t) e^{\frac{3 \iu \pi j}{N}} + \bar{B}(t) e^{-\frac{3\iu \pi j}{N}}.
\end{equation}
These details show that the non-degeneracy condition (\ref{non-degeneracy-1}) is satisfied for $N = 3, 4$.

\begin{ex}
	If $N = 3$, then $A$ is complex and $B$ is real. The system (\ref{pertshe}) in the decomposition (\ref{four-terms}) is reduced to the system of two equations
	$$
	\left\{ \begin{array}{l} 
	\dot{A} = \gamma A - 3 |A|^2 A - 12 B^2 A - 6 \bar{A}^2 B, \\
	\dot{B} = -(\lambda_3^N - \lambda_1^N)^2 B + \gamma B - \frac{1}{2} (A^3 + \bar{A}^3) - 6 |A|^2 B - 4 B^3.
	\end{array} \right.
	$$
	We compute the center manifold reduction $B = \Psi_3(A,\bar{A})$ in powers of $A$ according to (\ref{eqn-Psi-discrete}) with real $B$ by writing 
	\begin{align*}
	\Psi_3(A,\bar{A}) = \frac{1}{2} (A^3 + \bar{A}^3) \left[ c_0 + \mathcal{O}(|A|^2) \right],
	\end{align*}
	where $c_0$ is a real coefficient that depends on $\gamma$. 
	From the system of differential equations for $A$ and $B$, we find at the cubic order that 
	\begin{align*}
	c_0 &=  \frac{-1}{(\lambda_{3}^N - \lambda_{1}^N)^2 + 2 \gamma}, 
	\end{align*}
which agrees with the expansion (\ref{trunc-Psi}). Substituting $B = \Psi_3(A,\bar{A})$ into the first equation of the system, 
	we obtain consistently with (\ref{eqn-A-discrete}) that the slow dynamics of $A$ is given by 
	$$
	\dot{A}= A \left[ \gamma - 3 |A|^2 + \mathcal{O}(|A|^4) \right] 
	- 3c_0 \bar{A}^{5} \left[ 1 + \mathcal{O}(|A|^2) \right].
	$$
	Since $r_{N=3} = -3c_0 > 0$, the non-degeneracy condition (\ref{non-degeneracy-1}) is satisfied.
\end{ex}	
	
\begin{ex}
If $N = 4$, then both $A$ and $B$ are complex. The system (\ref{pertshe}) in the decomposition (\ref{four-terms}) is reduced to the system of two complex-valued equations
$$
\left\{ \begin{array}{l} 
\dot{A} = \gamma A - 3 (|A|^2 + 2 |B|^2) A - B^3 - 3 (\bar{A} B + \bar{B}^2) \bar{A}, \\
\dot{B} = -(\lambda_3^N - \lambda_1^N)^2 B + \gamma B - A^3 - 3 (2|A|^2 + |B|^2) B - 3 (A \bar{B}+ \bar{A}^2) \bar{B}.
\end{array} \right.
$$
We compute the center manifold reduction $B = \Psi_3(A,\bar{A})$ in powers of $A$ according to (\ref{eqn-Psi-discrete}) by writing 
\begin{align*}
\Psi_3(A,\bar{A}) = A^3 \left[ c_0 + c_1 |A|^2 + \mathcal{O}(|A|^4) \right] 
+ \bar{A}^{5} \left[ b_0 + \mathcal{O}(|A|^2) \right],
\end{align*}
where $c_0$, $c_1$, and $b_0$ are real coefficients that depend on $\gamma$. 
From the system of differential equations for $A$ and $B$, we find recursively at the cubic and quintic powers of $A$ that 
\begin{align*}
c_0 &=  \frac{-1}{(\lambda_{3}^N - \lambda_{1}^N)^2 + 2 \gamma}, \\
c_1 &=  \frac{-3}{[(\lambda_{3}^N - \lambda_{1}^N)^2 + 2 \gamma] 
	[(\lambda_{3}^N - \lambda_{1}^N)^2 + 4 \gamma]}, \\
b_0 &=  \frac{3}{[(\lambda_{3}^N - \lambda_{1}^N)^2 + 2 \gamma] 
	[(\lambda_{3}^N - \lambda_{1}^N)^2 + 4 \gamma]}.
\end{align*}
Substituting $B = \Psi_3(A,\bar{A})$ into the first equation of the system, 
we obtain consistently with (\ref{eqn-A-discrete}) that the slow dynamics of $A$ is given by 
$$
\dot{A}= A \left[ \gamma - 3 |A|^2 + \mathcal{O}(|A|^4) \right] 
-3 (b_0 + c_0^2) \bar{A}^{7} \left[ 1 + \mathcal{O}(|A|^2) \right].
$$
Since $r_{N=4} = -3(b_0 + c_0^2) < 0$, the non-degeneracy condition (\ref{non-degeneracy-1}) is satisfied.
\end{ex}

\begin{rem}
One can show with the method of mathematical induction that the non-degeneracy condition (\ref{non-degeneracy-1}) is satisfied for every $N \geq 3$. 
\end{rem}

The following lemma is important for the persistence argument, when the discrete SHE is perturbed by a small correction term.

\begin{lem}
Let $u_{\gamma}^N, v^N_{\gamma} \in \mathbb{R}^{\Z_N}$ be defined by Theorem \ref{theorem-SH-discrete} for $\gamma \in (0,\gamma_0)$ under the non-degeneracy condition (\ref{non-degeneracy-1}). Then, the matrix operator 
$$
\mathcal{A}_{\gamma} := -(A^N_W- \lambda^N_1)^2 + \gamma - 3 (u_{\gamma})^2 : \Z_N \to \Z_N,
$$ 
where $(u_{\gamma})^2$ is a diagonal matrix computed on the squared entries of $u_{\gamma}^N, v^N_{\gamma}$, is invertible.
\label{lem-SH-discrete}
\end{lem}

\begin{proof}
This follows from the available information about the eigenvalues of the linearized system (\ref{she-discrete}) 
at the time-independent solutions $\{ a_k\}_{k \in \Z_N}$ constructed from (\ref{eqn-A-discrete})
with (\ref{center-manifold}) and (\ref{eqn-Psi-discrete}). Eigenvalues of the linearized system (\ref{lin-eqs})
are bounded away from zero and so are eigenvalues $-(\lambda_k^N - \lambda_1^N)^2$ for $k \neq \pm 1$.
\end{proof}

\begin{rem}
	\label{rem-SH-discrete}
        The inverse matrix $\mathcal{A}_{\gamma}^{-1}$ in Lemma  \ref{lem-SH-discrete} behaves
        poorly as $\gamma \to 0$
        because the linearized system (\ref{lin-eqs}) 
has one eigenvalue $-2\gamma + \mathcal{O}(\gamma^2)$ and the other eigenvalue 
of $\mathcal{O}(\gamma^{N-1})$ under the non-degeneracy condition (\ref{non-degeneracy-1}).
As a result, $\| \mathcal{A}_{\gamma}^{-1} \| = \mathcal{O}(\gamma^{1-N}) \to \infty$ as $\gamma \to 0$.
\end{rem}

\section{The discrete SHE on $W$-random graphs}\label{sec.random}
\setcounter{equation}{0}

We now turn to the discrete SHE model (\ref{rshe}) on the $W$-random graph $\tl\Gamma^N_W$.
Using (\ref{def-L_N}), we rewrite it as follows
\begin{equation}
\label{Wnonshe}
\dot u =- \left(\tilde A^N_W - \tilde D^N_W -\kappa\right)^2 u + \gamma u - u^3, 
\qquad j \in \Z_N, 
\end{equation}
where $u\in C^1(\R, \R^{\Z_N})$ and
$\tilde A^N_W = (\tl a_{ij})$ is the adjacency matrix of $\tl\Gamma^N_W$. To make the setting for the model on a random graph consistent with the model on the deterministic Cayley graph, we have kept the periodic setting $\Z_N$ in (\ref{Wnonshe}). The matrix $\tilde{A}_W^N$ and the diagonal matrix $\tilde{D}_W^N$ acts on components of $u \in \Z_N$ inside $[-N, N-1]$ only.
The same notations $\tilde A_W^N$ and $\tilde D_W^N$ are used to denote the linear operators on $\Z_N$
and on $[-N, N-1]$.

We are interested in the Turing bifurcation of the trivial solution of \eqref{Wnonshe}.
Appendix~\ref{sec.app} shows that the matrix $\tilde{A}^N_W$ is close to the matrix $A^N_W$ in the cut-norm.
Consequently, eigenvalues of $\tilde{A}^N_W$ are close to the eigenvalues $\{ \lambda^N_k \}_{k\in \Z_N}$ of $A^N_W$
given by (\ref{eig-discrete}) \cite{Sze-2011}.
The following theorem presents the main result on the steady-state solutions
in (\ref{Wnonshe}).

\begin{thm}\label{thm.rand}
  Assume that $\lambda^N_k \neq \lambda^N_{1}$ for $k \neq \pm 1$ and choose $\kappa := \lambda^N_{1} - d^N_W = \lambda^N_{-1} - d_N^W$.
  Fix $\gamma \in (0,\gamma_0)$ with $\gamma_0$ given in Theorem \ref{theorem-SH-discrete} 
	and with $r_N \neq 0$ in (\ref{non-degeneracy-1}). There exist $N_0 \geq 3$ such that for every $N \geq N_0$ there exist at least $4$ and at most $4N$ values of $\delta$ such that the discrete SHE (\ref{Wnonshe}) admits time-independent solutions $u \in \R^{\Z_N}$ satisfying
	\begin{equation} 
	\label{Stokes-random}
	\sup_{j \in \Z_N} \left| u_j - \frac{2}{\sqrt{3}}  \sqrt{\gamma}
	\cos\left(\frac{\pi j}{N} - \delta \right) \right| \leq C_0 \sqrt{\gamma^3}.
	\end{equation}
Half of solutions are asymptotically stable in the time evolution of the discrete SHE (\ref{Wnonshe}) in $C^1(\R,\mathbb{R}^{\Z_N})$ and the other half of solutions are unstable.
 \label{thm.persistence}
\end{thm}

\begin{rem}
The system (\ref{Wnonshe}) does not have any symmetries except of the symmetry with respect to the sign reflection $u \mapsto -u$. Theorem~\ref{theorem-SH-discrete} gives existence
and stability of two distinct state $u_{\gamma}^N$ and $v_{\gamma}^N$ for sufficiently small
$\gamma > 0$, which are translated to every point of the lattice chain by the discrete translational
symmetry of Lemma \ref{obs.discrete}. Since the lattice chain in $\Z_N$ has $2N$ sites, we can count $4N$ distinct steady-state solutions in the discrete SHE model (\ref{pertshe}),
of
which $2N$ are stable and $2N$ are unstable. Compared to this conclusion, we do not have an exact count of the number of steady-state solutions on the random graph because of the broken symmetries. The number of steady solution is a random number divisible by 4 between 4 and 4N.
\end{rem}

\begin{rem}
Recall that the matrix $\mathcal{A}_{\gamma}$ in Lemma \ref{lem-SH-discrete} 
has a very small eigenvalue of the size $\mathcal{O}(\gamma^{N-1})$ for large $N \geq 3$.
Due to this small eigenvalue, we cannot not prove that the steady-state solutions of (\ref{pertshe})
persist as the steady-state solutions of (\ref{Wnonshe}) because the perturbation is not sufficiently small, see Appendix \ref{sec.app}, and the implicit function theorem cannot be used for the persistence argument. To overcome this problem,
we develop again the approach based on the center manifold reduction, where the main difference
is that the linear part of (\ref{Wnonshe}) is no longer diagonalizable by the discrete Fourier transform.
The cubic nonlinear term still enjoys the same transformation under the discrete Fourier transform as in
the proof of Theorem \ref{theorem-SH-discrete}. The other distinction from the deterministic setting
is that we can no longer use $\gamma > 0$
as a small continuation parameter. Instead, we consider the small parameter $\gamma$ as fixed in $(0,\gamma_0)$ with $\gamma_0$ given in 
Theorem \ref{theorem-SH-discrete} and continue the solution with respect to an additional small parameter $\mu$
induced by randomness which is only small for sufficiently large $N \geq N_0$.
\end{rem}

\begin{proof}[Proof of Theorem \ref{thm.persistence}]
As in the proof of Theorem~\ref{theorem-SH-discrete}, we use discrete Fourier transform
$$
u_j(t) =\sum_{k \in \Z_N} a_k(t) e^{\frac{ \iu \pi kj}{N}}, \quad j\in\Z_N.
$$
To apply the result in Appendix \ref{sec.app}, we write
$$
u=Fa, \quad F=\left(\omega^{jk}\right)_{-N\le j,k\le N-1}, \quad \omega=e^{\iu\pi/N},
$$
with $F^{-1}=(2N)^{-1} F^\ast$, where $F^\ast$ stands for the conjugate transpose of $F$.
By applying the inverse Fourier transform to both sides of \eqref{Wnonshe}, we have
	\begin{align}
	\notag
	\dot{a}_k &= -(\hat{L} a)_k + \gamma a_k 
	- \sum_{-N \leq k_1,k_2,k-k_1-k_2 \leq N-1} a_{k_1} a_{k_2} a_{k-k_1-k_2} \\
	\notag
	& \qquad \qquad - \sum_{-N \leq k_1,k_2,k+2N-k_1-k_2 \leq N-1} a_{k_1} a_{k_2} a_{k+2N-k_1-k_2} \\
	& \qquad \qquad - \sum_{-N \leq k_1,k_2,k-2N-k_1-k_2 \leq N-1} a_{k_1} a_{k_2} a_{k-2N-k_1-k_2},
	\label{she-random}
	\end{align}
	where
$$
\hat{L}= F^{-1} \left(\tilde A^N_W - \tilde D^N_W -\kappa\right)^2 F.
$$
Recall that similarity transformation $A^N_W \mapsto F^{-1} A^N_W F$ diagonalizes the
linear part of the deterministic model \eqref{pertshe}. Thus,
\begin{equation}\label{meet-hatL}
\hat{L} = \diag\{(\lambda_{-N}^N - \lambda_1^N)^2\}, \dots,
  (\lambda_{N-1}^N - \lambda_1^N)^2\} + \Delta,
  \end{equation}
where
\begin{equation}\label{meetDelta}
  \Delta = F^{-1} \left\{ ( (2N)^{-1} \tilde{A}^N_W - \tilde{D}^N_W -\kappa)^2 -((2N)^{-1} A^N_W - d^N_W -\kappa)^2\right\} F
  \end{equation}
  and we have used that $\kappa = \lambda_1^N - d_W^N$.  $\tl A$ is a symmetric matrix, whose 
entries above the main diagonal are
  independent random variables.
  Further,  $\E \tl A= A$.
  From these facts, by Bernstein's inequality, it follows  that with probability at least
  $1- \mathcal{O}\left(25^{-N}\right)$, we have 
\begin{equation}\label{key-estimate}
\max_{-N \leq j,k \leq N-1}  |\Delta_{jk} | \le C (\alpha^3_NN)^{-1/2}
\end{equation}
(cf.~Lemma~\ref{lem.Bern}).
Pick $\mu>0$ small and fixed. It follows from (\ref{key-estimate}) that with high probability for sufficiently large $N$, we have
  \begin{equation}\label{small-mu}
\max_{-N \leq j,k \leq N-1}  |\Delta_{jk} | \le \mu.
    \end{equation}

Proceeding along the lines of the proof of Theorem \ref{theorem-SH-discrete}, 
we express the slow manifold of the system (\ref{she-random}) as the graph 
of functions $a_0 = \Psi_0(A,\bar{A})$, $a_1 = A$, $a_{-1} = \bar{A}$, and 
	$$
	a_k = \Psi_k(A,\bar{A}), \quad 
	a_{-k} = \bar{\Psi}_k(A,\bar{A}), \quad k \in \{2,3,\dots,N \}.
	$$ 
Note that the difference from (\ref{center-manifold}) that the general $\hat{L}$ couples the odd-numbered and even-numbered Fourier amplitudes.

The functions $\Psi_k(A,\bar{A})$, $k \neq \pm 1$ can be obtained from the condition that $\gamma$ and $\mu$ are small, whereas $(\lambda_k^N - \lambda_1^N)^2$ in $\hat{L}_{kk}$ are strictly positive and bounded away from zero for every $k \neq \pm 1$. The dynamics on the slow manifold can be expressed by the amplitude equation:
\begin{equation}
\label{eqn-A-random}
\dot{A} = F_1(A,\bar{A}), \quad \dot{\bar{A}} = \overline{F_1(A,\bar{A})},
\end{equation}
where $F_1$ is a $C^{\infty}$ function in $A$ and $\bar{A}$ with ($\gamma$,$\mu$)-dependent coefficients. The Taylor series expansion 
includes only odd powers of $A$ and $\bar{A}$ due to the only symmetry 
of the random model (\ref{Wnonshe}) with respect to the sign reflection $u \mapsto -u$. 

The power expansion of $F_1(A,\bar{A})$ is different from those in the proof of Theorem \ref{theorem-SH-discrete} due to the presence of the perturbation terms 
in $\hat{L}_{\kappa}$ which are of the order of $\mathcal{O}(\mu)$. In addition, it is no longer true that $F_1(A,\bar{A})$ can be written in the form (\ref{eqn-A-discrete}). Nevertheless, the nonzero terms in the expansion 
of (\ref{eqn-A-discrete}) remain dominant terms in the expansion of 
(\ref{eqn-A-random}) if $\mu$ is sufficiently small. Thus, we have 
\begin{align}
F_1(A,\bar{A}) &= (\gamma + \mu \alpha_{1}) A + \mu \alpha_{2} \bar{A} \notag \\
& \quad + (-3 + \mu \beta_1) |A|^2 A + \mu \beta_2 A^3 + \mu \beta_3 |A|^2 \bar{A}+ \mu \beta_4 \bar{A}^3 + \dots \notag \\
& \quad   + [r_N + \mathcal{O}(\mu)] \bar{A}^{2N-1} + \dots \label{F-expansion}
\end{align}
with some $(\gamma,\mu)$-dependent coefficients which are bounded as 
$|\gamma| + |\mu| \to 0$. Since $\gamma \in (0,\gamma_0)$ is fixed and $\mu$ in (\ref{small-mu}) can be chosen sufficiently small, we have 
$\gamma + \mu \alpha_{1} > 0$, $-3 + \mu \beta_1 < 0$, and $r_N + \mathcal{O}(\mu) \neq 0$. By using the polar form $A = \rho e^{\iu \theta}$, we write
\begin{align}
\label{system-rho-theta}
	\left\{ \begin{array}{l} 
	\dot{\rho} = [\gamma + \mu \alpha_1 + \mu \alpha_2 \cos(2\theta)] \rho \\
	\qquad \qquad + 
	[-3 + \mu \beta_1 + \mu \beta_2 \cos(2\theta) + \mu \beta_3 \cos(2\theta) + \mu \beta_4 \cos(4 \theta)] \rho^3 + \mathcal{O}(\rho^5), \\
	\dot{\theta} = - \mu \alpha_2 \sin(2\theta) + \mu [ \beta_2 \sin(2\theta) - \beta_3 \sin(2\theta) - \beta_4 \sin(4 \theta) ] \rho^2 + \dots\\
	\qquad \qquad -[r_N + \mathcal{O}(\mu)] \rho^{2N-2} \sin(2N\theta) + \mathcal{O}(\rho^{2N}). \end{array} \right.
\end{align}
Since $\mu$ is selected to be much smaller than $\gamma$, there exists only one positive root of equation 
$$
\gamma + \mu \alpha_1 + \mu \alpha_2 \cos(2\theta) +
[-3 + \mu \beta_1 + \mu \beta_2 \cos(2\theta) + \mu \beta_3 \cos(2\theta) + \mu \beta_4 \cos(4 \theta)] \rho^2 + \mathcal{O}(\rho^4) = 0
$$
given by 
\begin{equation}
\label{roots-rho}
\rho = \frac{\sqrt{\gamma}}{\sqrt{3}} \left[ 1 + \mathcal{O}(\gamma,\mu) \right],
\end{equation}
independently of the value of $\theta$. The value of $\theta$ is defined from the second equation of system (\ref{system-rho-theta}) with $\rho^2 = \mathcal{O}(\gamma)$. Since there is no discrete translational symmetry of the system (\ref{Wnonshe}), compared to the system (\ref{pertshe}), we have to consider $\theta$ defined on $[0,2\pi)$. Roots of $\theta$ are defined from equation 
\begin{align*}
& - \mu \alpha_2 \sin(2\theta) + \mu [ \beta_2 \sin(2\theta) - \beta_3 \sin(2\theta) - \beta_4 \sin(4 \theta) ] \rho^2 + \dots\\
& \qquad \qquad -[r_N + \mathcal{O}(\mu)] \rho^{2N-2} \sin(2N\theta) + \mathcal{O}(\rho^{2N}) = 0,
\end{align*}
where $\rho$ is expressed from (\ref{roots-rho}). Since the left-hand-side is a trigonometric polynomial in $2 \theta$, there exist at least four roots of $\theta$ in $[0,2\pi)$ and at most $4N$ roots since $r_N + \mathcal{O}(\mu) \neq 0$. The total number of roots is divisible by $4$. For each root of $\theta$, the root of $\rho$ in (\ref{roots-rho}) is uniquely defined and the bound (\ref{Stokes-random}) with $\delta := \frac{N}{\pi} \theta$ follows.
	
The stability conclusion of Theorem \ref{thm.rand} follows from the linearization of system (\ref{system-rho-theta}) near each root and from 
the fact that all other eigenvalues of $-\hat{L}_{\kappa}$ are strictly negative 
of the order of $\mathcal{O}(1)$ for small $\gamma$ and $\mu$. 
Using notations $(\rho',\theta')$ for perturbation terms to the root $(\rho,\theta)$, we write 
the linearized equations in the form: 
		\begin{align*}
		\left\{ \begin{array}{l} 
	\dot{\rho'} = \displaystyle - [2 \gamma + \mathcal{O}(\gamma^2,\mu)] \rho' + \mathcal{O}(\sqrt{\gamma} \mu) \theta', \\
	\dot{\theta'} = \displaystyle \mathcal{O}(\sqrt{\gamma} \mu) \rho' 
	+ [- \mu \alpha_2 \cos(2 \theta) + \mathcal{O}(\gamma \mu) +  \left(\frac{\gamma}{3}\right)^{N-1} [r_N + \mathcal{O}(\gamma,\mu)] 2 N \cos(2N \theta)] \theta'. \end{array} \right.
	\end{align*}
Linearized evolution in $\rho'$ is asymptotically stable, whereas linearized evolution in $\theta'$ is asymptotically stable for half of solutions and is unstable for the other half of solutions.
\end{proof}

\begin{rem}
	Coefficients $\alpha_1$ and $\alpha_2$ in (\ref{F-expansion}) can be easily computed from the linear part of system (\ref{she-random}). Since $\Psi_k(A,\bar{A}) = \mathcal{O}(\mu) |A|$ for $k \neq \pm 1$, we have at the leading order:
	$$
	\mu \alpha_1 = \Delta_{1,1} + \mathcal{O}(\mu^2), \quad \mu \alpha_2 = \Delta_{1,-1} + \mathcal{O}(\mu^2).
	$$
\end{rem}

\begin{rem}
  If $\alpha_2 \neq 0$, we have exactly four time-independent solutions,
  from which two are asymptotically stable and two are unstable. The two stable (or unstable)
  solutions are related to each other by the sign reflection symmetry $u \mapsto -u$ of the discrete SHE
  (\ref{Wnonshe}). This follows from the fact that the center manifold reduction relies on a trigonometric polynomial
  in $2 \theta$ for which $\theta_0$ and $\pi + \theta_0$ are equivalent points. If $\Delta = 0$ in (\ref{meet-hatL}), then we have $4N$ time-independent solutions, identically to the outcome of Theorem \ref{theorem-SH-discrete}.
\end{rem}

\begin{rem}
	\label{rem-BraHol}
	In comparison with Theorem 4.1 in \cite{BraHol23}, we do not include 
	the quadratic terms in the discrete SHE models to avoid
        technical computations of the near-identity transformations.
        We also specify the particular case $k_0 = 1$ in Theorems \ref{theorem-SH-discrete} and
        \ref{thm.rand} to simplify computations
        of the normal forms. On the other hand, we give a precise statement of
        how the translational parameter $\delta$ of
        Theorem~\ref{theorem-SH-local} is determined in the case of the discrete graphs
        (both in the deterministic and random cases) and how many time-independent solutions
        exist for the discrete graph models.
        In addition, the proof of Theorem 4.1 in \cite{BraHol23} is incomplete. The analysis of the Fourier
        mode $w_2$ corresponding to the small eigenvalue $l_2 = -\delta^2 N^2 \rho_2$ is not
        included in the proof. In our setting, the dynamics of $w_2$ 
        is captured by the equation for $\theta$ in \eqref{system-rho-theta}.
        This equation is important, because it determines the number of
        branches bifurcating from the spatially homogeneous equilibrium.
\end{rem}

\section{The discrete SHE on small-world graphs}
\label{sec.small-world}
\setcounter{equation}{0}

In this section, we illustrate the bifurcation analysis of the discrete SHE models on the deterministic and random graphs with
numerical results.
To this end, we use the family of small-world graphs from Example~\ref{ex.SWgraphon}. This is a representative example,
for which Assumption~\ref{ass-bifurcation} can be verified analytically.

Recall the definition of the small-world graphon $W(x,y)=S(x-y)$ with $S\in L^1(\T)$  given by 
\begin{equation}\label{def-S}
S(x)=\left\{ \begin{array}{ll}
               1-p, & |x|\le r,\\
               p,& r<|x|\le \frac{1}{2},
             \end{array}
           \right.
           \end{equation}
where $p \in [0,1]$ and $r \in \left(0,\frac{1}{2}\right)$. The eigenvalues of 
the Hilbert-Schmidt operator $K_S$ in \eqref{def-K_S} are known explicitly
(cf.~\cite{ChiMed19a}):
	\begin{equation}\label{EVsK}
	\lambda_k=\left\{\begin{array}{ll}
	2r (1-2p) +p, & k=0,\\
	(\pi k)^{-1} (1-2p)\sin\left(2\pi k r\right),& k \in \mathbb{Z} \backslash \{0\}.
	\end{array}
      \right.
      \end{equation}
Eigenvalues for $k = 0$ and $k \in \mathbb{N} \backslash \{1\}$ are shown in Figure \ref{f.eigenvalues} by blue dots. The red dot shows the bifurcating eigenvalue at $k_0 = 1$ for which we select $\kappa = \lambda_1 - d_S$ with $d_S = \lambda_0$. It is clear from Figure \ref{f.eigenvalues} that $\lambda_k < \lambda_1$ for $k \geq 2$ so that Assumption \ref{ass-bifurcation} is satisfied.

\begin{figure}[htb!]
  \centering
\includegraphics[width=10cm,height=8cm]{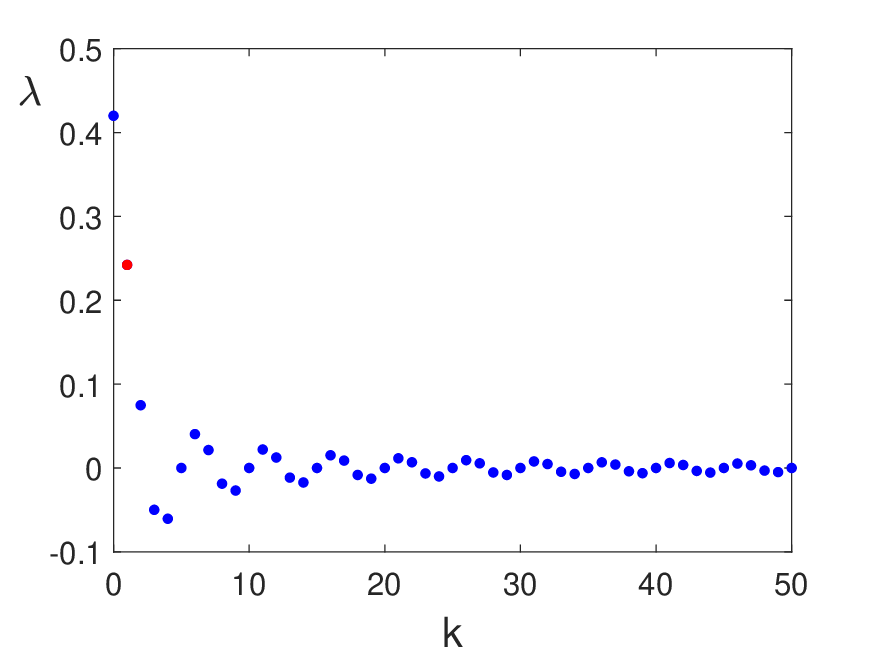}
\caption{Eigenvalues of $K_S$ given by (\ref{EVsK}) for $N = 50$, $p=0.1$ and $r=0.2$.}
	\label{f.eigenvalues}
\end{figure}

By Theorem~\ref{thm.Turing}, for small $\gamma>0$, 
the continuous SHE model (\ref{nonshe}) has a continuous 
family of asymptotically orbitally stable solutions $\{ u_{\gamma}(\cdot + \delta)\}_{\delta \in \mathbb{T}}$, where $u_{\gamma}$ is approximated by 
\begin{equation}
\label{stable-continuous}
      u_{\gamma}(x)= 2\sqrt{\frac{\gamma}{3}} \cos(2\pi x) + {\mathcal O}(\gamma^{3/2}).
\end{equation}

We shall now consider how the stable solutions with the expansion (\ref{stable-continuous}) persist
in the discrete SHE models on the deterministic and random graphs. For the discrete SHE model (\ref{pertshe})
with  (\ref{def-LN}) and \eqref{def-S}, we compute the eigenvalues of $A^N_W$ in the form
      $$
      \lambda_k^N =\frac{1-p}{2N} \sum_{|j|\le \lceil rN\rceil} e^{\frac{-\iu \pi k j}{N}}
      +\frac{p}{2N} \sum_{|j|> \lceil rN\rceil} e^{\frac{-\iu \pi k j}{N}}.
      $$
We set $\kappa = \lambda^N_1-d_W^N$ with $d^N_W = \lambda_0^N$. For small $\gamma>0$, Theorem~\ref{thm.discrete-Turing} yields existence of the two discrete families of solutions $\{ \sigma_m u^N_{\gamma}  \}_{m \in \Z_N}$ and $\{ \sigma_m v^N_{\gamma} \}_{m \in \Z_N}$, 
where $\sigma_m$ is the shift operator defined by 
$(\sigma_m u)_j = u_{j+m}$, $j \in \Z_N$ and the profiles of $u^N_{\gamma}$ and $v^N_{\gamma}$ are approximated respectively by 
\begin{align} 
	\label{stable-discrete}
u_j & =\frac{2}{\sqrt{3}}  \sqrt{\gamma}
                 \cos\left(\frac{\pi j }{N}\right)  + \mathcal{O}\left(\gamma^{3/2}\right),\\
 \label{unstable-discrete}
v_j &=  \frac{2}{\sqrt{3}}  \sqrt{\gamma}
	\cos\left(\frac{\pi j}{N} - \frac{\pi}{2N} \right) + \mathcal{O}\left(\gamma^{3/2}\right), 
	\end{align}
for $j \in\Z_N$. One of the two solutions is asymptotically stable and the other solution is unstable.

The random SHE model (\ref{Wnonshe}), we use the same $\kappa$ and define 
the symmetric matrix $\tilde A^N_W$ with zeros on the main diagonal and with  the entries above the main diagonal $\tilde a^N_{ij}$ being independent random variables such that
$$
\P(\tilde a^N_{ij}=1) =a^N_{ij},\qquad\mbox{and}\qquad  \P(\tilde a^N_{ij}=0) =1-a^N_{ij},
$$
where $a^N_{ij}$ is a coefficient of the adjacency matrix of the weighted Cayley graph.
Theorem~\ref{thm.persistence} implies the existence of asymptotically stable solutions with the approximation
\begin{align} 
\label{stable-random}
u_{\gamma,\delta} =\frac{2}{\sqrt{3}}  \sqrt{\gamma}
\cos\left(\frac{\pi j}{N} - \delta \right) + \mathcal{O}\left(\gamma^{3/2}\right)
\end{align}
for small $\gamma>0$ and appropriate fixed $\delta$.

\begin{figure}[htb!]
	\centering
   \textbf{a}~\includegraphics[width=6cm,height=5cm]{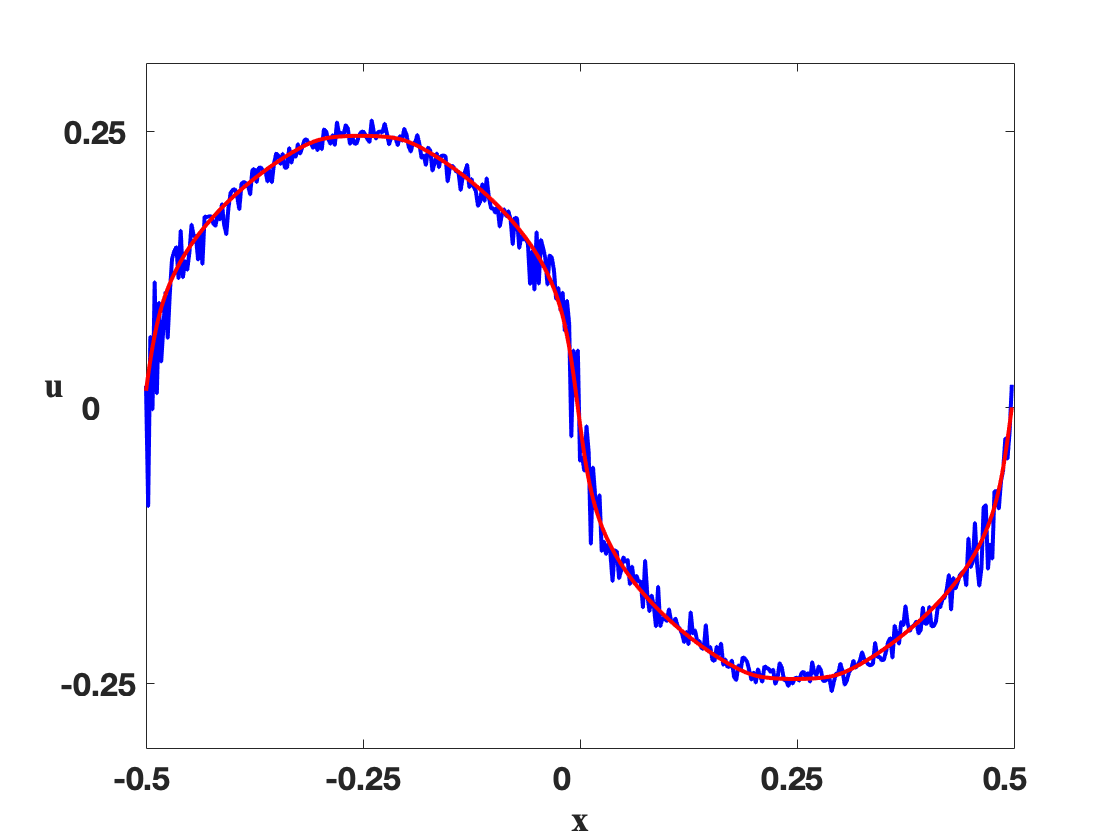}
   \textbf{b}~\includegraphics[width=6cm,height=5cm]{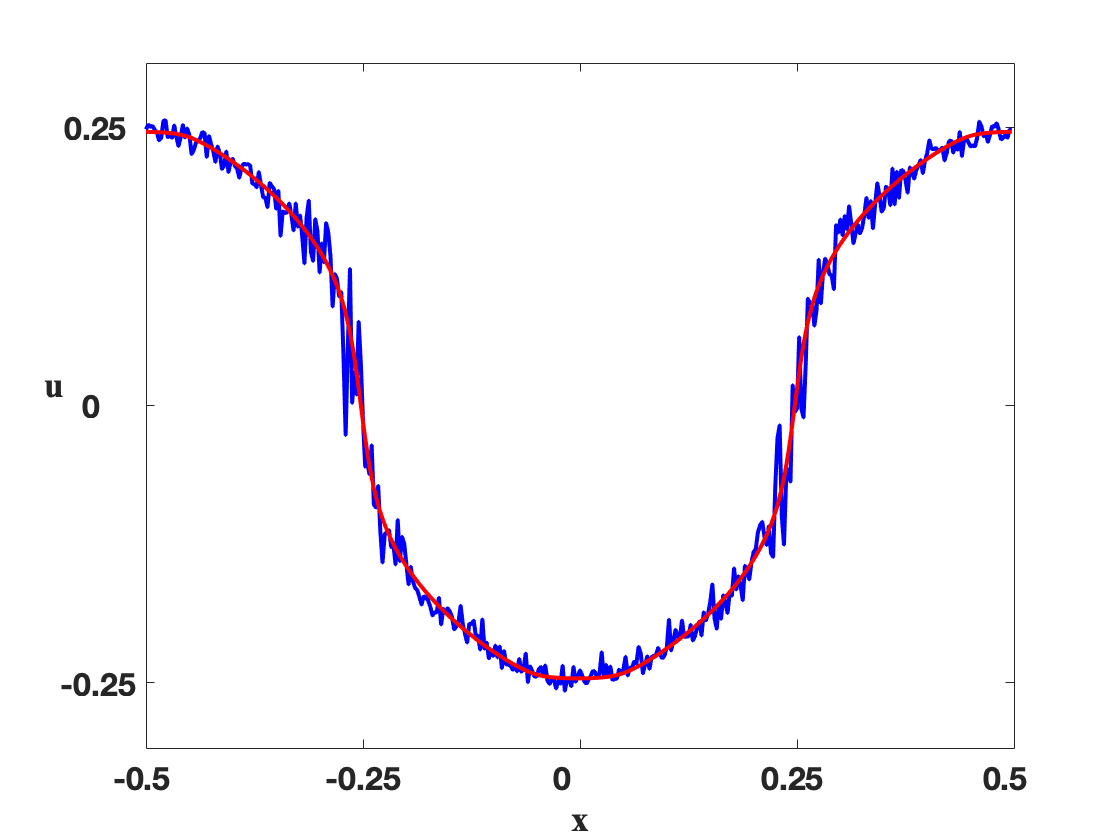}
               \caption{Numerical solutions of the SHE on deterministic SHE with kernel \eqref{def-S} with parameters
                  $N=400$, $p=0.1$, $r=0.2$, and $\gamma=0.05$ plotted in red and its random counterpart plotted in blue.
                  Both models were initialized with the discretization of the leading order term in \eqref{stable-continuous}
                  and integrated for 100 units of time. The shift $\delta$ was set to $-\pi/2$ in plot \textbf{a} and to $0$
                  in plot \textbf{b}.
                }
	\label{f.patterns}
\end{figure}

Figure~\ref{f.patterns} presents results of numerical simulations of the discrete models derived from the
nonlocal SHE with kernel \eqref{def-S} with parameters
$N=400$, $p=0.1$, $r=0.2$, and $\gamma=0.05.$
Both the deterministic and random models were initialized with the leading order term on the right-hand side
of \eqref{stable-continuous}. The shift $\delta$ was set to $-\pi/2$ in plot \textbf{a} and to $0$ in plot \textbf{b}.
The two models were integrated for $100$ units of time. The solutions of the deterministic models are plotted in red
and the solutions of the random model are plotted in blue. 

On finite time intervals, the solutions of the discrete SHE on deterministic and random graphs are
expected to remain close, provided the initial data are sufficiently close and $N$ is large (cf.~\cite{Med19}).
This is clearly seen in the simulations shown in Figure~\ref{f.patterns}. The snapshots of the trajectories
of the deterministic  and the random models, shown in red and blue respectively, lie very close to each other.
For the deterministic model, we know that the trajectory relaxes to one of the $2N$ stable states lying in the vicinity of the initial condition. As to the snapshots of the SHE on random graph, we see that they
relax to the form predicted by Theorem~\ref{thm.rand}. We do not know whether the solutions of the
random SHE shown in Figure~\ref{f.patterns} are close to the actual equilibria, because the evolution
in the translational direction is extremely slow due to the $O(\gamma^{N-1})$ eigenvalue, and cannot be resolved
by numerics for large $N$.


\begin{figure}[htb!]
	\centering
\textbf{a}~	\includegraphics[width=7.5cm,height=6cm]{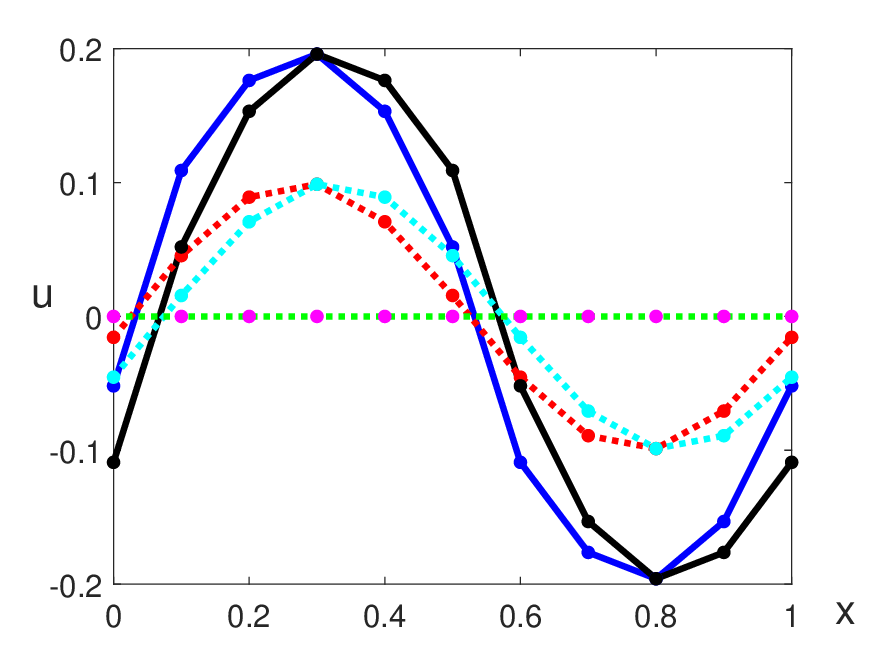}
\textbf{b}~\includegraphics[width=7.5cm,height=6cm]{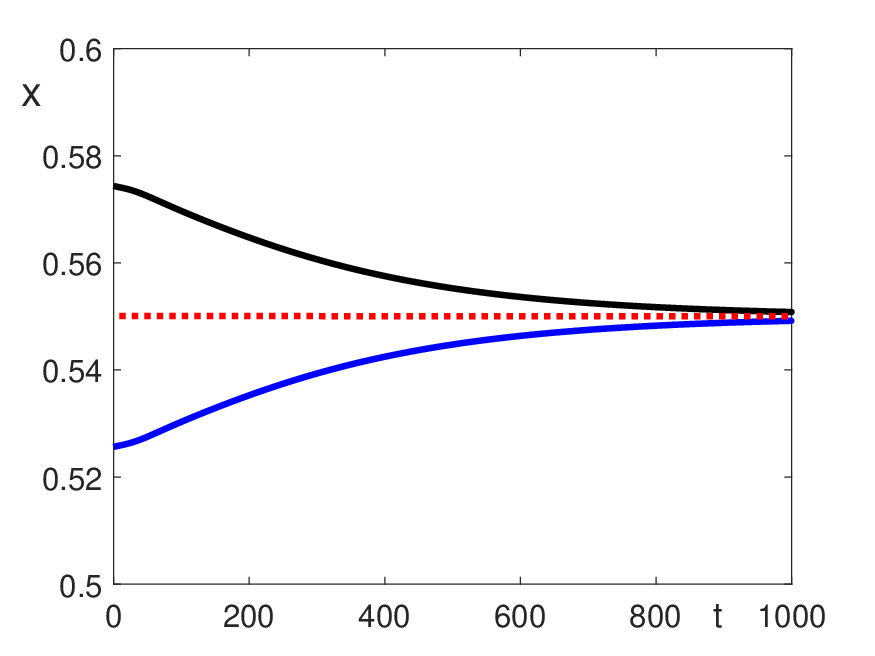}
	\caption{Numerical approximations of the  discrete SHE model (\ref{pertshe})  for $N=5$, $p = 0.9$, $r = 0.2$, and $\gamma = 0.03$.    \textbf{a}~ Two initial conditions (red and cyan curves) and the outcomes of numerical simulations at $t = 100$ (blue and black curves).    \textbf{b}~ Numerically obtained zeros of the two solutions versus time relative to the midpoint between two grid points (red dotted line).}
	\label{f.determ}
\end{figure}

To illustrate the effects of the slow drift in numerical simulations,  we have performed simulations of the
deterministic and random SHE models for $N = 5$.
Figure \ref{f.determ} shows outcomes of the numerical simulations for the discrete SHE model (\ref{pertshe})
with two different initial conditions given by the red and cyan lines. The two initial data are given by
(\ref{stable-continuous}) with a half of the amplitude and shifted to the right and to the left relative to the
grid points (shown by magenta dots). The two solutions quickly converge to the near-equilibrium solutions
with the correct amplitude of $0.2$ (blue and black curves) which then slowly shift further from the grid points
towards the mid-point between the two grid points. The snapshot on the left panel was stopped at time $t = 100$,
but the slow drift of the zeros of the two solutions on the much longer time interval is shown on the right panel.
Zeros are computed from the linear interpolation between the grid points and the color scheme is the same as
that for the final states on the left panel. This numerical experiment shows that the asymptotically stable solution is given
by a discrete translation of the solution $v^N_{\gamma}$ given by the expansion (\ref{unstable-discrete}). 

\begin{figure}[htb!]
	\centering
   \textbf{a}~	\includegraphics[width=7.5cm,height=6cm]{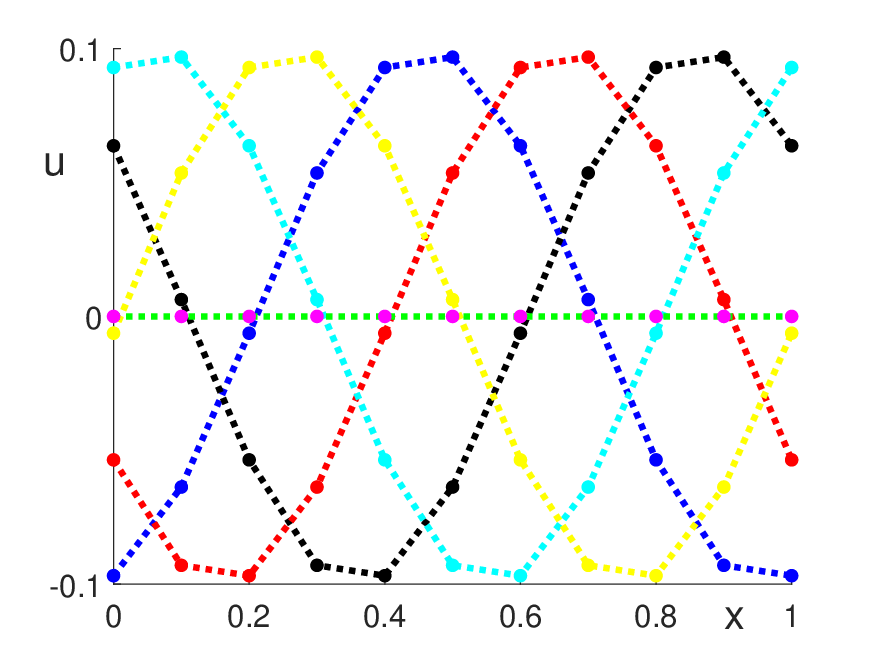}
   \textbf{b}~	\includegraphics[width=7.5cm,height=6cm]{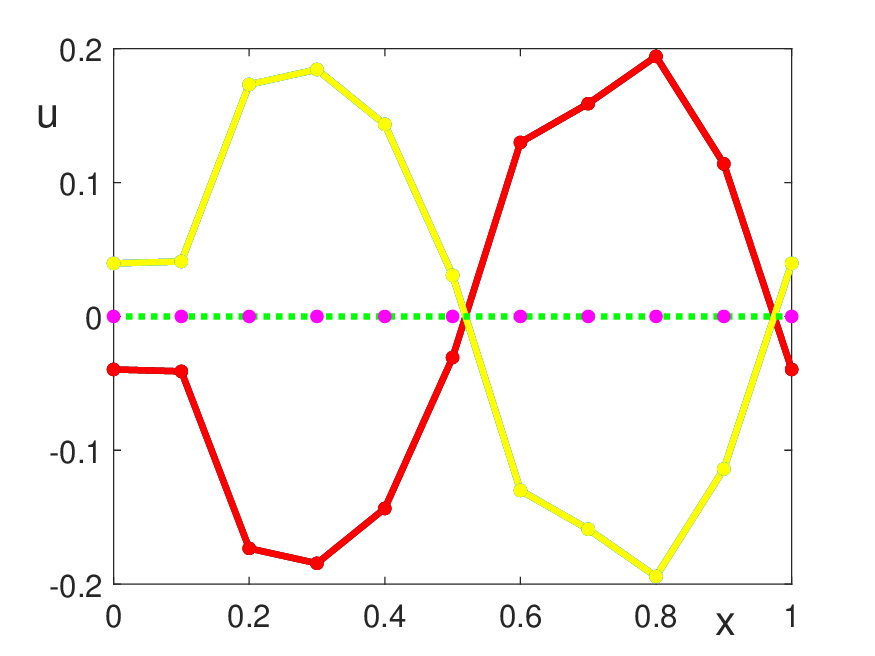} \\
   \textbf{c}~	\includegraphics[width=7.5cm,height=6cm]{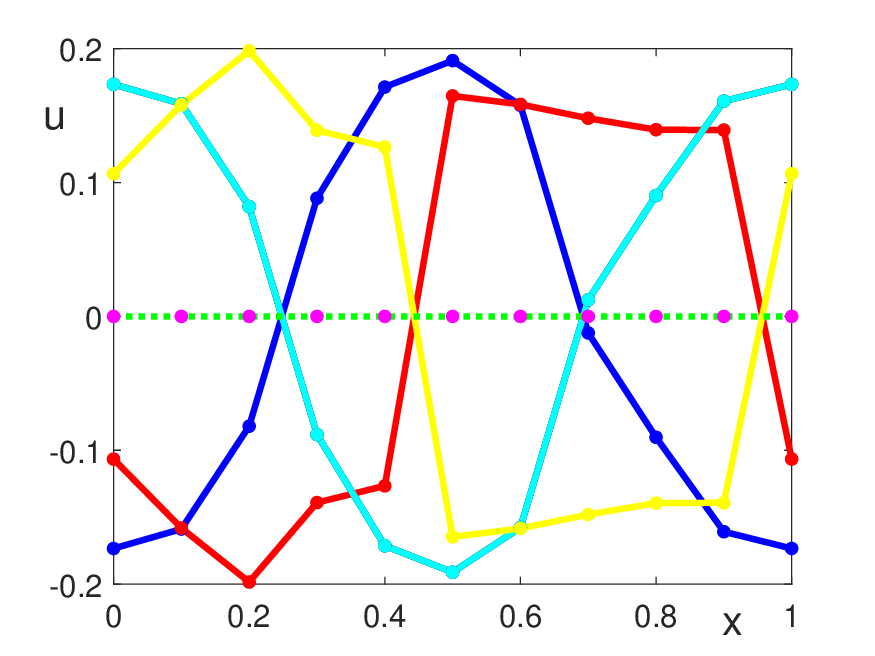} 
   \textbf{d}~	\includegraphics[width=7.5cm,height=6cm]{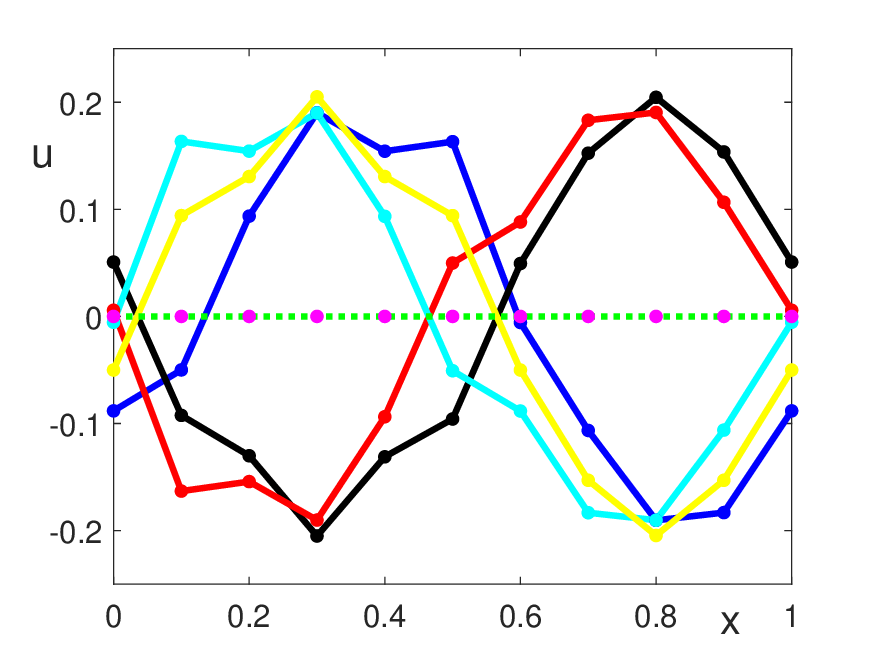} 
	\caption{Numerical approximations of random realizations of the discrete SHE model (\ref{Wnonshe}) for $N=5$, $p = 0.9$, $r = 0.2$, and $\gamma = 0.03$.    \textbf{a}~ Five initial conditions shifted relative to each other. 
   \textbf{b}-\textbf{d}~ Outcomes of numerical simulations with two, four, and six stable solutions. }
	\label{f.random}
\end{figure}

Figure \ref{f.random} shows outcomes of the numerical simulations for the random realizations of the discrete SHE
model (\ref{Wnonshe}) with five different initial conditions shown on the top left panel. The drift generally occurs
faster in the random model because the relevant small eigenvalue is of the size $\mathcal{O}(\mu)$ if $\alpha_2 \neq 0$,
where $\mu$ is defined by (\ref{small-mu}). Only two time-independent solutions are stable if $\alpha_2 \neq 0$
and this generic case is shown on the top right panel. The two stable solutions are given by (\ref{stable-random})
with some specific $\delta$ and they are related by the sign reflection of each other. The color scheme of the final
states corresponds to that of the initial conditions and the missing colors correspond to the final states identical to those shown in red and yellow.

If $\alpha_2$ is zero or close to zero, then four, six, eight, or ten time-independent solutions could be stable according to
Theorem \ref{thm.rand} and since half states are related by the sign reflection of the other half states,
five numerical conditions could be used to detect all cases. However, the cases with more stable states
become rare events. Our numerical simulations showed random configuration of the discrete SHE model
with four and six stable states (bottom left and right panels, respectively) but we did not find random
configurations with eight and ten stable states. The four configurations on the bottom left panel shows
that two states are given by the sign reflection of the other two states (red versus yellow and blue versus cyan).
The same is true for the bottom right panel (red versus cyan and black versus yellow) but there is one more stable
state (blue line) which is not matched by the sign reflected state (since our simulations only involve five initial conditions). 

 \section{Discussion}\label{sec.discuss}
     \setcounter{equation}{0}
     
     Dynamical principles underlying formation of patterns in complex systems are important for understanding
     a range of phenomena arising across multiple disciplines from morphogenesis to autocatalytic reactions, to
     firing patterns in neuronal networks \cite{Murray-Bio}. Traditionally, pattern formation has
     been studied
     in the continuous setting through the framework of reaction-diffusion partial differential equations.
     Recently, the  interest has shifted towards understanding patterns in discrete systems, driven by the
     ubiquity of networks in contemporary science.
     Interestingly, already in his seminal paper on morphogenesis, Turing considered a reaction-diffusion
     model on a lattice, i.e., a discrete system \cite{Turing-Morph}. There is an extensive literature on lattice
     dynamical systems, which covers pattern formation and propagation phenomena \cite{HanKloeden-Lattice}.
     In the past  decade, there have been many studies exploring pattern formation in
     complex networks including random networks. The key analytical challenge in dealing with this class of models,
     which was not present in studies of partial differential equations  or lattice dynamical systems, is
     handling network topology, which can be random.
     Until recently most studies of the dynamics of complex and random
     networks had to resort to heuristic and numerical arguments
     (see \cite{NakMikh2010, Wol2012, HuttAmb2022, AsllaniBusiello, KouHatGui2015} for a representative albeit
      limited sample of studies of Turing patterns in networks).

      The situation has changed with the development of the theory of graphons  \cite{LovGraphLim12}.
      The use of graphons allows a rigorous
     derivation of the continuum limit for interacting dynamical systems on a large class of graphs
     including random
     graphs \cite{Med14a, Med14b, Med19}, which can be used effectively for studying dynamics on
     large networks.
     For example, the use of graphons led to the breakthrough in the analysis of synchronization
     and pattern formation in systems of coupled phase oscillators on networks \cite{CMM23, MedMiz22},
     interacting diffusions on graphs \cite{ORS18, Luc2020},
     mean-field games \cite{CaiHua21}, and graph signal processing \cite{Ruiz2020GraphonSP, GJK22}.

     Graphons provide multiple analytical benefits. Oftentimes, graph limits possess additional
     symmetries that are not present in the individual realizations of random graphs. For example, the graph limit
     of the small-world family of graphs used in the present work is isotropic, in contrast to the graphons
     corresponding
     to the realizations of small-world graphs. This symmetry enables the effective use of the Fourier transform for
     analyzing the limiting model (cf.~Section~\ref{sec.continuum}).
     Likewise, the deterministic (averaged) discrete model \eqref{pertshe} is shift invariant and can be studied
     using discrete Fourier transform (cf.~Section~\ref{sec.Cayley}). This in turn can be used to understand the
     dynamics of the random model
     \eqref{Wnonshe}. In general, the proximity in cut norm of a network to a symmetric network in
     can facilitate the analysis of the original nonsymmetric model.

     In the context of the bifurcation problems considered in this paper, the relationship between the
     spectra and eigenspaces of the deterministic and random discrete models, as well as their continuum
     counterpart, is crucial. The theory of graph limits offers very efficient tools for tracking this relation. The
     convergence of kernel operators in cut--norm automatically implies the proximity of the corresponding
     eigenvalues and eigenspaces \cite{Sze-2011}. This has significant implications for the analysis of
     dynamical systems.
     For instance, once the proximity in cut norm of  the linear operators corresponding to the interaction
     terms of the
     random model and its deterministic counterpart was verified (cf.~Appendix~A), the proximity
     of the eigenvalues and the corresponding eigenspaces followed automatically.
     In contrast, the analysis in
     \cite{BraHol23} based on operator norm topology and Davis-Kahane estimates requires substantial
     efforts. Importantly, the analysis in the present paper extends to sparse networks.

     Our techniques    apply naturally to other pattern--forming
     systems on random graphs, including Gierer--Meinhardt model \cite{HuttAmb2022},
     Mimura--Murray model of interacting prey-predator populations \cite{NakMikh2010} and many other
     activator--inhibitor systems. An interesting area of  potential applications  are 
      neural fields \cite{CGP14}. Fourier methods played
      an important role in the analysis of pattern in nonlocal neural field models, which are 
      very similar to the continuum limit analyzed in Section~\ref{sec.continuum} \cite{LaiTroy2003}.
      We expect that the methods developed this paper
     may lead to interesting results for discrete neural fields with random connectivity. Finally, 
     the concentration estimate for the random linearized operator in Fourier space (cf.~Appendix A)
     may be useful in graph signal processing.

      \appendix
      \section{A concentration inequality}\label{sec.app}
        \setcounter{equation}{0}
        
        Let $A=(a_{ij})\in\R^{n\times n}$ be a symmetric matrix such that $ 0\le a_{ij}\le 1$ and
        $
        a_{ii}=0, \; i\in [n],
        $
        and let  $\tl A=(\tl a_{ij}) \in\R^{n\times n}$ such that $\tl a_{ij},\; 1\le i< j\le n$ are independent
        random variables defined as follows
        $$
       \P (\tl a_{ij}=\alpha_n^{-1})=\alpha_na_{ij}, \qquad \P (\tl a_{ij}=0)=1-\alpha_na_{ij},
        $$
        and $\tl a_{ij}=a_{ji}, \; \tl a_{ii}=0$. Here, $\alpha_n$ is a positive sequence satisfying
        \begin{equation}\label{repeat-alpha}
  1\ge \alpha_n\ge Mn^{-1/3},
  \end{equation}
for some $M>0$ dependent of $N$. Further, define
        \begin{align*}
          D=\diag(d_1,d_2,\dots, d_n), &\quad d_i=n^{-1} \sum_{j=1}^n a_{ij},\\
          \tl D=\diag(\tl d_1,\tl d_2,\dots, \tl d_n), &\quad \tl d_i=n^{-1} \sum_{j=1}^n \tl a_{ij},
        \end{align*}
        and
        $\bA= n^{-1} A$, $\tl\bA=n^{-1} A.$ Consider
        \begin{equation}\label{def-Delta}
        \Delta = F^{-1} \left\{ (\bA-D-\kappa)^2 -(\tl\bA -\tl D-\kappa)^2\right\} F,
      \end{equation}
        where
        \begin{equation*}
F=\left(\omega^{-(j-1)(k-1)}\right)_{1\le j,k\le n},\qquad \omega\doteq
e^{\frac{ \iu 2\pi}{n}},\qquad F^{-1}=n^{-1} F^\ast.
\end{equation*}
The main result of this appendix is the following lemma with the straightforward corollary.

\begin{lem}\label{lem.Bern}
  With probability at least $1-O\left(5^{-n}\right)$,
  $$
\max_{1 \leq j,k \leq n} |\Delta_{jk}| \le C (\alpha_n^3n)^{-1/2},
  $$
  where $C$ does not depend on $n$.
\end{lem}
\begin{cor}\label{cor.Bern}
  For a given $\epsilon>0$ with high probability, for all sufficiently large $n$,
  $$
\max_{1 \leq j,k \leq n} |\Delta_{jk}|  \le \epsilon,
  $$
  provided $M$ in \eqref{repeat-alpha} is large enough.
 \end{cor}

We precede the proof of Lemma \ref{lem.Bern} with a few comments. By construction of $\tl A$,
$\tl a_{ij},\; 1\le i< j\le n$ are independent
random variables, $\E\tl A=\E A$, and
\begin{equation}\label{def-p}
  p\doteq \frac{2}{n(n-1)} \sum_{i<j} \operatorname{Var}{\alpha_n\tl a_{ij}} =\frac{2}{n(n-1)} \sum_{i<j}  \alpha_n a_{ij}
  \left(1-\alpha_na_{ij}\right).
\end{equation}
Note that 
\begin{equation*}
p\simeq \left\{ \begin{array}{ll}\frac{1}{2}\int_Q W(1-W)dx+o(1), & \alpha_n\equiv 1 \qquad \mbox{(dense case)},\\
                  \alpha_n\frac{1}{2}\int_Q Wdx +o(1), &   \alpha_n\searrow 0\qquad \mbox{(sparse case)}.
                \end{array}
                \right.
\end{equation*}
In either case,
\begin{equation}\label{p-asymptotics}
  C_1 \alpha_n\le p\le C_2 \alpha_n
  \end{equation}
for appropriate positive $C_1, C_2$ independent of $n$.

  The proof of Lemma \ref{lem.Bern} relies on the estimates of 
  $\cut{\tl\bA -\bA}$ and $\cut{\tl D -D}$ based on the Bernstein inequality
  (cf.~\cite[Theorem~4.3]{GueVersh2016}).
  Here, the $\infty\to 1$ norm of $A\in\R^{n\times n}$ is defined as follows
\begin{equation}\label{def-cut}
  \cut{A} \doteq \max_{x,y\in \{-1,1\}^n} \left| \sum_{i,j=1}^n a_{ij} x_i y_j\right|.
\end{equation}
  In particular, for $X\in \{\bA, D\}$, we have
\begin{equation}\label{GueVer}
  \cut{\tl X -X}\le  3 \alpha_n^{-1} p^{1/2} n^{3/2} \simeq \alpha_n^{-1/2} n^{3/2}.
\end{equation}
holding with probability at least $1-e^3 5^{-n}$ provided
$$
p>\frac{9}{n}.
$$
For $X=\bA$ \eqref{GueVer} follows from Lemma~4.1 in \cite{GueVersh2016} and \eqref{p-asymptotics}
For $X= D$ \eqref{GueVer} is proved by following the same steps as in the
proof of  Lemma~4.1 in \cite{GueVersh2016}.

\begin{proof}[Proof of Lemma~\ref{lem.Bern}]
  We rewrite  \eqref{def-Delta} as follows
  \begin{align}\nonumber
    \Delta
           & = F^{-1}\left( \tl\bA^2 -\bA^2\right) F-2\kappa F^{-1} \left( \tl\bA -\bA\right) F\\
     \label{6terms}
           &-  F^{-1} \left( \tl\bA -\bA\right) \tl D F -  F^{-1} \bA \left( \tl D - D\right) F\\
    \nonumber
    &- F^{-1}D \left( \tl\bA -\bA\right) F -  F^{-1} \left( \tl D - D\right)\bA F.
  \end{align}

  Denote the six terms in the order as they appear on the right hand side of \eqref{6terms} by
  $F^{-1} S_i F,$ $i\in [6]$. We claim for each $i\in [6]$,
  \begin{equation}\label{claim}
    \left\| F^{-1} S_i F\right\|_\infty \le C (\alpha_n^3 n)^{-1/2}
    \end{equation}
    holding with probability $1-O(5^{-n})$. Once we verify \eqref{claim}, the proof will be complete.
    We verify \eqref{claim} only for $i=1,2$, as the remaining terms are estimated in the same
    manner.

    We start with the second term on the right hand side of \eqref{6terms}
    \begin{align*}
      \left|\left(F^{-1} \left( \tl\bA -\bA\right) F\right)_{ij}\right| & =n^{-2}
\left| \sum_{l,k} F_{il} \left( \tl A - A\right)_{lk} F_{kj}\right|
  \le 4 n^{-2} \max_{x,y\in\{-1,1\}^n} \left| \sum_{l,k} x_l \left( \tl A - A\right)_{lk} y_k\right|\\
           &=n^{-2} \|A-\tl A\|_{\infty\to 1} \le C (\alpha_n n)^{-1/2},
\end{align*}
where we used  $F^{-1}=n F^\ast$, $|F_{kl}|=1$, and \eqref{GueVer}.

We now turn to the first term
$$
F^{-1}\left( \tl\bA^2 -\bA^2\right) F= n^{-3} F^\ast \left( \tl A^2 - A^2\right) F
$$
In this case, we need to estimate,
  \begin{equation} \label{new-matrix}
     n^{-3} F^\ast \left( \tl A^2 - A^2\right) F =
     n^{-3} F^\ast \tl A \left( \tl A - A\right) F + n^{-3} F^\ast \left( \tl A - A\right) A F.
  \end{equation}
  We estimate the first term on the right hand side of \eqref{new-matrix}. The second term is dealt with similarly.
  Let
   \begin{equation} \label{new-term}
     E=n^{-3} F^\ast \tl A \left( \tl A - A\right) F= \alpha_n^{-1} n^{-3} F^\ast (\alpha_n\tl A) \left( \tl A - A\right) F
   \end{equation}
   Denote
   \begin{align*}
     x& =\operatorname{col}_i \left( n^{-1} \alpha_n \tl A F\right)= n^{-1} (\alpha_n\tl A) \operatorname{col}_i (F^\ast),\\
     y& =\operatorname{col}_j (F).
   \end{align*}
   Note $\|x\|_\infty\le 1$.
   Thus,
   \begin{align*}
     \left|E_{ij}\right| & = \alpha_n^{-1} n^{-2} \left|\sum_{i,j}  x_i  \left( \tl A - A\right)_{ij} y_j\right|\\
                         & \le 4\alpha_n^{-1} n^{-2} \left\| \tl A - A\right\|_{\infty\rightarrow 1}
                           \le C \frac{1}{(\alpha_n^3 n)^{1/2}}.
   \end{align*}

   This completes the analysis of $S_1$, the first term on the right hand side of \eqref{6terms}. The remaining terms
   are analyzed similarly and result in either $O\left((\alpha_nn)^{-1/2}\right)$ bound as for $S_2$ above or
   in $O\left((\alpha_n^3 n)^{1/2}\right)$ bound for $S_1$. Thus, $\|\Delta\|_\infty=O\left((\alpha_n^3 n)^{1/2}\right)$
   as claimed.
     \end{proof}

   \begin{rem}\label{rem.final} Note that $O\left((\alpha_n^3 n)^{1/2}\right)$ come from quadratic terms like $S_2$.
     In the second order nonlocal spatial operator in \eqref{nonshe} is replaced with the first order operator
     \eqref{def-K_S}, as one encounters in the neural field type models, then the bound on $\|\Delta\|_\infty$ can
     be improved to $O\left((\alpha_n n)^{1/2}\right)$. This means that the results of this paper would hold for
     $\alpha_n=O(n^{-1})$, i.e., for graphs of \textit{bounded degree}.
     \end{rem}

  \vskip 0.2cm
\noindent
{\bf Acknowledgements.} This work was partially supported by
NSF grant DMS 2009233 (to GSM) and NSERC Discovery grant (to DEP).

\bibliographystyle{amsplain}

\def\cprime{$'$} \def\cprime{$'$} \def\cprime{$'$}
\providecommand{\bysame}{\leavevmode\hbox to3em{\hrulefill}\thinspace}
\providecommand{\MR}{\relax\ifhmode\unskip\space\fi MR }
\providecommand{\MRhref}[2]{%
  \href{http://www.ams.org/mathscinet-getitem?mr=#1}{#2}
}
\providecommand{\href}[2]{#2}

\end{document}